\documentclass[mnsc,nonblindrev]{informs3}

\let\theoremstyle\relax
\usepackage{amsthm}
\usepackage{hyperref}
\usepackage{cleveref}

\makeatletter
\def\thm@space@setup{%
  \thm@preskip=0.5em plus 0.25em minus 0.1em
  \thm@postskip=\thm@preskip 
}
\makeatother

\theoremstyle{plain}%
\newtheorem{theorem}{Theorem}[section]
\newtheorem{lemma}[theorem]{Lemma}
\newtheorem{proposition}[theorem]{Proposition}

\theoremstyle{definition}

\newtheorem{example}[theorem]{Example}

\usepackage{nicefrac}

\usepackage{enumitem}
\setlist[description]{font=\normalfont\bfseries\space}

\OneAndAHalfSpacedXI

\usepackage{microtype}

\usepackage{natbib}
 \bibpunct[, ]{(}{)}{,}{a}{}{,}%
 \def\bibfont{\small}%
 \def\bibsep{\smallskipamount}%
 \def\newblock{\ }%
\setcitestyle{authoryear}

\usepackage{mathtools}
\usepackage{thmtools}
\usepackage{svg}
\usepackage{verbatim}
\usepackage{wrapfig}

\DeclarePairedDelimiter{\abs}{\lvert}{\rvert}

\DeclarePairedDelimiter{\paren}{\lparen}{\rparen}



\newcommand*{\R}{\mathbb{R}}   
\newcommand*{\GF}{\mathbb{F}}  

\makeatletter
\newcommand*\rel@kern[1]{\kern#1\dimexpr\macc@kerna}
\newcommand*\widebar[1]{%
  \begingroup
  \def\mathaccent##1##2{%
    \rel@kern{0.8}%
    \overline{\rel@kern{-0.8}\macc@nucleus\rel@kern{0.2}}%
    \rel@kern{-0.2}%
  }%
  \macc@depth\@ne
  \let\math@bgroup\@empty \let\math@egroup\macc@set@skewchar
  \mathsurround\z@ \frozen@everymath{\mathgroup\macc@group\relax}%
  \macc@set@skewchar\relax
  \let\mathaccentV\macc@nested@a
  \macc@nested@a\relax111{#1}%
  \endgroup
}
\makeatother
\newcommand*{\Men}{\mathcal{M}}    
\newcommand*{\Women}{\mathcal{W}}  
\newcommand*{\Types}{\Theta}       
\newcommand*{\men}{m}              
\newcommand*{\women}{w}            
\newcommand*{\type}{\theta}        
\newcommand*{\typeii}{{\theta'}}   
\newcommand*{\ntypes}{n}           

\newcommand*{\utility}{u}          
\newcommand*{\utilityii}{u'}          
\newcommand*{\contutility}{\widebar u} 
\newcommand*{\thresh}{{\tau}}    
\newcommand*{\threshii}{\tau'}    
\newcommand*{\CDF}{F}              
\newcommand*{\Dist}{\mathcal{D}}   

\newcommand*{\mass}{\eta}          
\newcommand*{\rate}{\lambda}       
\newcommand*{\glow}{f}             

\newcommand*{\strategy}{\sigma}
\newcommand*{\strategyii}{\sigma'}

\newcommand*{\OPT}{\mathsf{OPT}}   
\newcommand*{\nfill}{n}
\newcommand*{\tee}{t}

\newcommand*{\exeps}{\varepsilon}

\newcommand*{\neqs}{m}
\newcommand*{\neqsii}{m'}

\newcommand*{\nvars}{n}
\newcommand*{\var}{x}
\newcommand*{\bit}{b}
\newcommand*{\switch}{s}
\newcommand*{\switchval}{2(1 + 2\delta)}
\newcommand*{\eqn}{e}
\newcommand*{\MaxThreeLinTwo}{\mathsf{MAX3LIN2}}
\newcommand*{\PlatformOpt}{\mathsf{OptimalDirectedSearch}}
\newcommand*{\NP}{\mathsf{NP}}
\newcommand*{\approxerr}{\varepsilon}

\newcommand*{\Waiting}{\mathsf{Waiting}}
\newcommand*{\Exited}{\mathsf{Exited}}
\newcommand*{\Meeting}{\mathsf{Meeting}}
\newcommand*{\PreMeeting}{\mathsf{PreMeeting}}
\newcommand*{\Accept}{\mathsf{Accept}}
\newcommand*{\Reject}{\mathsf{Reject}}

\newcommand*{\weighting}{q}
\usepackage{hyperref}

\usepackage{comment}
 \usepackage[suppress]{color-edits}
\addauthor{vm}{blue}
\addauthor{aw}{violet}

\begin{document}

\TITLE{Designing Approximately Optimal Search on Matching Platforms}

\RUNAUTHOR{Immorlica, Lucier, Manshadi, Wei}

\RUNTITLE{Designing Approximately Optimal Search}

\ARTICLEAUTHORS{%
\AUTHOR{Nicole Immorlica}
\AFF{Microsoft Research, New York, NY, \EMAIL{nicimm@gmail.com}}
\AUTHOR{Brendan Lucier}
\AFF{Microsoft Research, Cambridge, MA, \EMAIL{brlucier@microsoft.com}}
\AUTHOR{Vahideh Manshadi}
\AFF{Yale School of Management, New Haven, CT, \EMAIL{vahideh.manshadi@yale.edu}}
\AUTHOR{Alexander Wei }
\AFF{UC Berkeley, Berkeley, CA, \EMAIL{awei@berkeley.edu}}
}

\ABSTRACT{

We study the design of a decentralized two-sided matching market in which agents' search is guided by the platform.
There are finitely many agent types, each with (potentially random) preferences drawn from known type-specific distributions.
Equipped with knowledge of these distributions, the platform 
guides the search process by determining the meeting rate between each pair of types from the two sides. 
Focusing on symmetric pairwise preferences in a continuum model, we first characterize the unique stationary equilibrium that arises given a feasible set of meeting rates. 
We then introduce the platform's 
optimal directed search problem, which involves optimizing meeting rates to maximize equilibrium social welfare.
We first show that incentive issues arising from congestion and cannibalization 
make the design problem fairly intricate.
Nonetheless, we develop an efficiently computable search design whose corresponding equilibrium achieves at least $\nicefrac{1}{4}$ the 
social welfare of the optimal design.
In fact, our construction always recovers at least $\nicefrac{1}{4}$ the \emph{first-best} social welfare, where agents' incentives are disregarded.
Our directed search design is simple and easy-to-implement, as its corresponding bipartite graph consists of disjoint stars.
Furthermore, our design implies the platform can substantially limit choice and yet induce an equilibrium with an approximately optimal welfare.
Finally, we show that approximation is likely the best we can hope for by establishing that the problem of designing optimal directed search is $\mathsf{NP}$-hard to even approximate beyond a certain constant factor.
%
}

\maketitle

\pagenumbering{arabic}

\section{Introduction}\label{sec:introduction}

\awedit{Matching platforms have become prominent facilitators of social and economic connections in recent years: Thirty percent of U.S. adults have used dating platforms to look for a partner \citep{pew2020}; and 35\% of U.S. workers have engaged in some form of freelance labor in the past year, in part thanks to the growth of online marketplaces for freelance work \citep{ozimek2019}. 
Platforms, such as e-Harmony, try to improve search via match recommendations that are
based on learned user characteristics and match compatibility models.
While carefully designed models have predictive power, the compatibility of a pair often also involves an {\em idiosyncratic} component that can only be discovered upon meeting. Furthermore, since matching requires a ``coincidence of wants,'' the platform cannot compel users to match, even if preferences are known.  In such environments, how can matching platforms improve the search process while respecting users' incentives?}


Toward answering this question, 
we construct a dynamic matching market 
between two heterogeneous populations, each composed of finitely many types. 
To be concrete, we take a dating market for heterosexual couples as our base example and thus refer to agents of the two sides as women and men.
An agent's type captures their common features observable to the platform and partially determines their preferences for opposite-side agents.\footnote{We remark that platforms such as e-Harmony extract these features using surveys; however, we abstract away from such details and assume that these features are directly observable.} 
In addition to the type-specific component, each agent's preference for any opposite-side agent has an idiosyncratic component that is unknown a priori and is only revealed when the pair meets.\footnote{Such preference structures have been studied previously in matching literature; see e.g., \citet{ashlagi2020clearing} and \citet{KanoriaSaban21}.}
In our model, agents have cardinal preferences, and each woman-man pair shares a symmetric valuation that is drawn from a distribution corresponding to their types. 
Symmetric preferences capture the notion of mutual compatibility in a relationship, i.e., any relationship is either win-win or lose-lose.
Upon meeting, both agents observe their shared valuation; each then decides whether to match or to wait for another candidate. 
We consider a continuum  model where agents of different types arrive at exogenous rates and leave the market either upon meeting a satisfactory match or unmatched due to a life event which occurs at a given rate (see  \Cref{sec:dynamics} and \Cref{fig:flows}). Such an exogenous departure rate captures {\em search friction} in the sense that an agent can only meet a limited number of candidates before leaving unmatched due to a life event. 

With knowledge of preference distributions and exogenous arrival/unmatched departure rates, the platform guides the search process by designing \emph{meeting rates} between pairs of woman-man types. The set of meeting rates for a given type can be viewed as an \emph{assortment} of opposite-side types offered to them over time. 
Faced with such an assortment, an agent decides on which candidates to accept or reject in order to maximize their long-run utility. We focus on the stationary and symmetric equilibria of the underlying dynamic game: We show for any set of ``feasible'' meeting rates that respect natural physical constraints (see \cref{eq:capacity,eq:flowbalance}), there exists a unique such equilibrium in which each agent type plays a threshold strategy (see \Cref{proposition:unique-main}). 
Establishing the uniqueness of the stationary equilibrium relies on the symmetric structure of preferences. In fact, in our  constructive proof, we show that iterating the best response map converges after finitely many rounds (see \Cref{proposition:unique}).
Uniqueness of equilibrium along with its simple structure enables us to sidestep issues of instability and equilibrium selection which may make the platform's design problem ill-defined.
As such, our structural results may be of independent interest in future work on the design of dynamic matching markets.\footnote{We highlight that several previous papers considered symmetric preference structures in static settings under other names such as {\em correlated two-sided} or {\em acyclic} markets \citep{abraham2008stable, ackermann2011uncoordinated}.} 



We remark that by taking such a design approach to search, our work departs from the prevailing search environment studied in the literature \citep{BurdettColes97, shimer2000assortative, adachi2003search, LauermannNoldeke14}, which assumes agents meet uniformly at random. 
In the presence of sufficient differentiation across types, such a ``hands-off'' approach to search can be arbitrarily sub-optimal: 
If the platform already knows that 
sports fans only marry sports fans and outdoor enthusiasts only marry outdoor enthusiasts, it should not waste time by letting outdoor enthusiasts meet sports fans. (For a quantitative example, see \vmedit{the horizontal market presented in} \Cref{example:horizontal}.) 
At the same time, the optimal meeting design can be far more nuanced than simply letting preferred pairs meet each other. This is due to disparities in arrival rates combined with strategic behavior of agents, which can lead to issues of cannibalization and congestion. 
To alleviate these issues, sometimes the platform may wish to restrict the choices of a type in order to induce a matching outcome with globally higher welfare. 
To illustrate these behaviors, in \Cref{sec:examples}, we present an instance of an assortative market with deterministic preferences for which the optimal meeting design is indeed non-assortative 
\vmedit{(see the vertical market presented in \Cref{example:vertical}}). 

The aforementioned behaviors suggest that optimal design can be fairly subtle.
In fact, we show that the platform's  optimal directed search problem (formally introduced in \Cref{sec:optimaldirectedsearch}) is $\mathsf{APX}$-hard, i.e., $\NP$-hard to approximate beyond a fixed constant factor (see \Cref{theorem:hardness-main}). 
Nonetheless, we develop a solution (i.e., a set of feasible meeting rates) whose equilibrium welfare is at least $1/4$ of that of the optimal design (see \Cref{theorem:approximation}). Our design relies on first relaxing the incentive constraints and solving the first best problem. Through an intricate reformulation, we show that the first-best problem can be reduced to a polynomial-time solvable ``generalized assignment problem'' that admits a forest-structured solution (see \Cref{sec:reduction}).
We then show that the assignment problem admits a $2$-approximate solution consisting of disjoint ``star-shaped'' submarkets, i.e., submarkets where one side consists entirely of a single type. Finally, we recompute the first-best for each submarket; utilizing its star structure, we show that the first-best solution can be ``modified'' to respect incentive constraints with at most a factor of $2$ loss in welfare (see \Cref{sec:starshaped}, \Cref{proposition:starshaped}, and \Cref{fig:outline}). 
This holds because for a star-shaped market, there is an alignment of incentives between individual agents and the platform.

\vmedit{We complement our theoretical developments with numerical studies of markets where agents' preferences have both a vertical and a horizontal component (see Section \ref{subsec:setup}). Our comparative statics with respect to the strength of each component as well as the intensity of search friction illustrate that, compared to random meeting, our design is particularly effective in markets with strong horizontal components or high search friction (see \Cref{fig:experiment1,,fig:experiment2} in Section \ref{sec:simulation}). Furthermore, our simulation results show that the welfare outcome of our approximation algorithm can be substantially better than its theoretical guarantee.}

\awdelete{The simple structure of our proposed solution---a collection of disjoint star-shaped markets---implies that the platforms can substantially limit choice, offer very simple assortments, and yet induce an equilibrium with an approximately optimal welfare. Our result is particularly intriguing as we do not impose any restrictions on the structure of symmetric preferences. 
Furthermore, limiting choice makes the agent's decision problem, i.e., finding their equilibrium threshold, substantially easier.}

\awedit{The simple structure of our proposed solution—a collection of disjoint star-shaped markets— implies that, with careful search design, platforms can substantially limit choice, offer very simple assortments, and yet induce an equilibrium with approximately optimal welfare. By limiting choice, our proposed solution also significantly simplifies the agent's decision problem and thus reduces the need for complex strategization. Our result is particularly intriguing since it holds without any restrictions on the structure of the symmetric preferences.}

\subsection{Related Work}
\label{subsec:lit:review}

Our work relates to and contributes to several streams of literature on matching markets. 

\paragraph{Search and Matching.}
There is a rich literature in economics that studies decentralized matching models with search frictions and transferable/nontransferable utility. For an informative review, we refer the interested reader to \citet{chade2017sorting}. By and large, this literature focuses on uniformly random meeting among agents and studies the properties of the resulting stationary equilibria. \vmedit{
Our framework is partly inspired by the seminal works of
\citet{BurdettColes97} and \citet{smith2006marriage} in this literature. These papers consider matching models with nontransferable utility in which agents of each side belong to a continuum of types.
In the work of \citet{BurdettColes97}, an agent's match utility only depends on the type (pizazz) of the other agent, while \citet{smith2006marriage} considers more general pair-specific utilities (production) that may be symmetric (akin to our preference structure). 
However, unlike these papers, we do not assume that match utility is deterministic and vertically differentiated.\footnote{\vmedit{\cite{chade2006matching} extends the model of~\citet{BurdettColes97} to incorporate uncertainty in preferences; however, their model differs crucially from ours in that, even upon meeting, only a noisy version of the utility is observed.}} This implies that, in our general model, types are incomparable and  assortative matching is not well-defined.\footnote{\vmedit{After our Example~\ref{example:vertical}, which considers a vertically differentiated market, we further discuss the connection between this example and one of the  examples presented in~\citet{BurdettColes97}.}} (See Example~\ref{example:horizontal} for a horizontally differentiated market within our framework.) }


\vmedit{Another line of work in this literature studies the relationship between the stationary equilibria of decentralized search under uniformly random meeting and corresponding sets of stable matchings.}
Under a ``cloning'' assumption---which keeps the distribution of agents unchanged by asserting that each time two agents match and leave, two single agents identical to them join---\citet{adachi2003search} shows that as search friction vanishes, the set of stationary equilibria converges to the set of stable matchings of a corresponding centralized market. 
Relaxing such a cloning assumption, and thus allowing for the distribution of agents to be endogenously determined, \citet{LauermannNoldeke14} show that stationary equilibria converge to stable matchings if and only if there exists a unique stable matching in a corresponding discrete market. 

\vmedit{Following this literature, 
we also model search friction as a time consuming process of meeting partners and abstract away from incorporating 
any potential second-stage decision or cost
upon meeting (e.g., whether to go on a date with a potential partner). 
However, we emphasize that} we complement this literature by taking a design approach---motivated by online matching platforms---to optimize
directed search (i.e., non-uniform and pairwise type-specific meetings) in a model with cardinal and symmetric preferences. \vmedit{In the spirit of designing the search environment, \citet{shimer2001matching} study a model with transferable utility in which search is costly and search intensities can vary across types. (However, agents still meet uniformly at random within the searching population.) They show that a socially optimal solution (i.e., first-best) can arise as an equilibrium of a decentralized search process with 
a linear tax or subsidy on search intensity.}

\paragraph{Directed Search and Platform Interventions.}
Moving beyond random meeting, a series of recent papers has studied different forms of platform design and intervention to facilitate search. 
\citet{halaburda2016competing} focus on the impact of limiting choice and show that when agents' outside options are heterogeneous, a platform that offers limited choice and charges a higher price can still compete with ones without any choice restriction. 
A recent work of \citet{KanoriaSaban21} shows that limiting the action of agents, e.g., allowing for only one side to propose, and hiding information about the quality of some agents can lead to welfare improvement. 
 \vmedit{We now further discuss some of the key differences between our paper and \citet{KanoriaSaban21}. First, we remark that the set of actions of agents in \citet{KanoriaSaban21} is richer in that, in addition to deciding on whether to match upon meeting, agents decide on whether to request meeting a candidate and whether to inspect that candidate (at a cost).
While the model of \citet{KanoriaSaban21} also includes exogenous departure rates, their paper focuses on the regime where departure rates approach zero. As a result, each agent has effectively unboundedly many meeting opportunities and the main search friction is inspection cost. 
In contrast, in our model, each agent has a finite number of meeting opportunities; upon each meeting---after observing the match utility---the agent decides whether to accept the match or to risk departing unmatched before meeting another candidate. (As mentioned above, a majority of papers in the search literature take a similar modeling approach.)
Another key difference is the nature of the platform intervention. \citet{KanoriaSaban21} investigate blocking one side from proposing and hiding information about quality in the special case of a vertical market where one side consists of a single type and the other one has two (high/low) types. In contrast, we focus on designing meeting distributions---inspired by platform recommendation systems---for general (not necessarily vertical) markets. Overall,} compared to the \vmedit{aforementioned} papers, our model and our intervention are more fine-grained, in that both sides can have multiple \vmedit{(potentially incomparable)} types and we design pairwise type-specific meeting rates.


The approach of \citet{banerjee2017segmenting} to designing visibility graphs in a two-sided market with different types of buyers and sellers bears some similarity with ours. However, there are several key differences: We consider a matching market without transfers in which the platform's directed search design impacts agents' acceptance thresholds on both sides, while 
\citet{banerjee2017segmenting} focus on a network of buyers and sellers exchanging a single undifferentiated good in which the platform impacts clearing prices by choosing a visibility subgraph. 

Motivated by the role of matching platforms in shaping agent's choice, a sequence of recent papers \citep{ashlagi2019assortment, aouad2020online, rios2020improving} consider the problem of assortment planning in two-sided matching in both static and dynamic settings. 
In this novel line of work, agents are non-strategic in that their behavior is captured by a choice model such as the Multinomial Logit model. Consequently, agents are oblivious to the action of those on the other side and the platform design. While sharing a similar motivation with this emerging literature, our work complements it by capturing agents' strategic behavior in response to how the platform's design guides the search process. 
In our work, the set of meeting rates for a type can be viewed as an assortment of options that the platform (sequentially) offers to an agent type. The equilibrium response to the collection of such assortments determines the matching and welfare outcome.


\paragraph{Matching with Incomplete Information.}
Outside the framework of search, several papers explore the problem of finding a stable matching where preferences are a priori unknown. 
Taking a communication complexity approach, \citet{gonczarowski2019stable} and \citet{ashlagi2020clearing} establish bounds on the the amount of communication, measured by the number of bits, needed to find a stable match in markets with private preferences. The recent work of \citet{immorlica2020information} focuses on a setting where learning preferences is costly and show how costly information acquisition impacts an agent's preference.
Furthermore, a few recent papers, such as \citet{liu2020competing}, use the multi-armed bandit framework to model the process of learning preferences as  an online learning problem and develop efficient learning algorithms.  Finally, in another direction, \citet{emamjomeh2020complexity} analyze an iterative query process for learning a stable matching under general preferences.
\vmedit{A closely related line of work focuses on specific matching algorithms such as the Men-Proposing Deferred Acceptance Algorithm, and examines agents' incentives to manipulate  preferences (see for example \citet{roth1982economics}, \citet{coles2014optimal}, and  \citet{immorlica2015incentives}).}

\paragraph{Matching in Dynamic Environments. }
\vmedit{
A growing number of papers study the design of dynamic matching markets motivated by various applications ranging from school choice \citep{feigenbaum2020dynamic} to ride-sharing \citep{nikzad2017thickness} to public housing \citep{afeche2021optimal}. In this body of literature, several papers focus on the timing (or frequency) of matching in a centralized market and the trade-off between market thickness and waiting time \citep{AshlagiKE, ashlagi2019matching, akbarpour2020thickness}. In a closely related direction, papers such as \citet{doval2019efficiency} and \citet{arnosti2020design} study the setting where strategic agents have heterogeneous preferences and thus face a trade-off between matching with a less preferred choice (including their outside option) or waiting longer. Recently, several papers examine the role of information design on managing congestion in dynamic allocation problems \citep{anunrojwong2021information, ashlagi2021optimal, che2021optimal}.  
Our work complements the aforementioned papers by studying the design of a dynamic matching market in which strategic agents with heterogeneous preferences repeatedly meet potential matches recommended by the platform. 
}







\section{Model}\label{sec:model}

In this section, we introduce a search model in a dynamic, two-sided matching market, where the search process is facilitated by recommendations from a centralized platform. A graphical depiction of our model, focused on the dynamics around a single type, is shown in \Cref{fig:flows}.

\ifx\acmConference\undefined
    \begin{figure}
        \FIGURE
        {\includesvg[scale=0.4]{figures/flows.svg}}
        {This schematic depicts the dynamics for a single type $\type$. Each dark gray box represents a type, with the size corresponding to the type's population mass. The light gray connections between boxes correspond to how agents meet under directed search. The colored lines depict inflow and outflow for type $\type$, with width corresponding to flow rate: Blue for arrivals, green for matching, and red for leaving unmatched. In stationary equilibrium, the inflow should balance the outflow (as is depicted).\label{fig:flows}}{}
    \end{figure}
\else
    \clearpage
    \begin{wrapfigure}{r}{0.5\textwidth}
        \vspace{-3em}
        \centering
        \includesvg[scale=0.4]{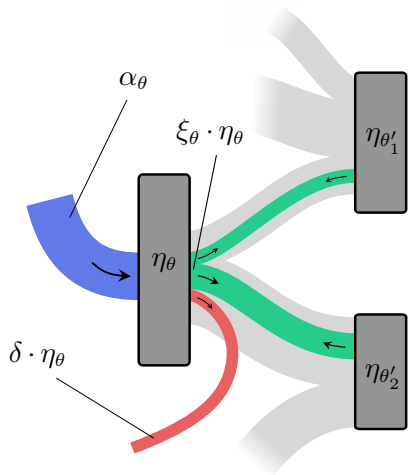}
        \caption{This schematic depicts the dynamics for a single type $\type$. Each dark gray box represents a type, with the size corresponding to the type's population mass. The light gray connections between boxes correspond to how agents meet under directed search. The colored lines depict inflow and outflow for type $\type$, with width corresponding to flow rate: Blue for arrivals, green for matching, and red for leaving unmatched. In stationary equilibrium, the inflow should balance the outflow (as is depicted).}
        \label{fig:flows}
    \end{wrapfigure}
\fi

\subsection{Agents, Types, and Dynamics}\label{sec:dynamics}

We model our matching market as a flow economy, where the market consists of a continuum of agents of infinitesimal mass.
Agents on the two sides are differentiated into finitely many types
belonging to the sets $\Men$ and $\Women$, respectively.
Agents' types determine prior distributions over their preferences (see \Cref{sec:utilities}). We use $\Types = \Men\cup\Women$ to denote the combined type space of all agents. For each type $\type\in\Types$, we denote 
the population mass of type $\type$ agents present on the platform by
$\mass_\type$.
Throughout, we will use the notational convetion that $\typeii$ refers to a generic type on the opposite side of the market from type $\type$.

Our matching market takes place in continuous time. To track change in the population masses $\mass_\type$ over time, we use \emph{flow rates}---the mass of agents arriving or departing per unit time.\footnote{Although $\mass_\type$ can in general vary over time, we suppress the time index on $\mass_\type$ for notational compactness and because our eventual focus will be on stationary equilibria, in which case $\mass_\type$ will be constant over time.} We assume entry into the market is fixed exogenously, with type $\type$ agents entering at a constant flow rate of $\alpha_\type$ (i.e., an $\alpha_\type$ mass of type $\type$ agents enters per unit time). This corresponds to the blue inflow in \Cref{fig:flows}. On the other hand, we let departure from the market be partially determined endogenously, with all agents eventually either leaving with a match or experiencing an exogenous ``life event'' that causes them to leave unmatched. More specifically, the two ways an agent can depart are as follows:
\begin{description}
    \item[Matching.] 
    Agents leave the market when they enter into a mutually agreed upon match with another agent (see the green outflow in \Cref{fig:flows}).
    We define $\xi_\type$ so that $\xi_\type\mass_\type$ is the flow rate at which type $\type$ agents do so. That is, each \emph{individual} of type $\type$ (assuming symmetry between agents of type $\type$) leaves matched with probability $\xi_\type\,d\tee$ during the infinitesimally small time interval $[\tee, \tee + d\tee)$. The value of $\xi_\type$ is determined both by how the platform recommends potential matches and who the agents themselves choose to match with (see \Cref{lemma:rate}).\footnote{Note that $\xi_\type$ can also vary over time; we suppress the time index as we did for $\mass_\type$.}
    \item[Life event.] Agents also leave unmatched when they experience a ``life event'' (see the red outflow in \Cref{fig:flows}). We assume life events occur randomly for each individual at a constant rate of $\delta\,d\tee$. That is, each individual experiences a life event with probability $\delta\,d\tee$ during each infinitesimally small time interval $[\tee, \tee + d\tee)$. (Equivalently, the time an agent spends in the market is drawn (unknown to them) from an exponential distribution of rate $\delta$.) At the individual level, the possibility of leaving unmatched acts as a discount factor and ensures that agents do not search forever. At the platform level, individuals leave at a total flow rate of $\delta\mass_\type$, which ensures that the mass of unmatched agents does not grow unboundedly.
\end{description}

Our focus will be on the stationary equilibria of this market, where the inflow of arriving agents balances the outflow of departing agents. We say that \emph{stationarity} holds if 
\begin{equation}\label{eq:stationarity}
    \alpha_\type = (\delta + \xi_\type)\mass_\type
\end{equation}
for all $\type\in\Types$. Note that $\delta + \xi_\type$ is the total rate at which \emph{individual} type $\type$ agents exit the market.

\subsection{Preferences and Utilities}\label{sec:utilities}

We assume that agents have cardinal preferences (i.e., a utility associated to matching with each potential partner) and normalize the value of leaving unmatched to $0$.
In our model, the utility values of specific potential partners are \emph{unknown} ex ante. Instead, pairs of agents learn their valuations for each other when they meet while searching (see \Cref{sec:directedsearch}), with their utilities being drawn from a prior distribution given by their types.\footnote{We assume that agents do not meet twice, since for any agent, there are uncountably many agents of each type that they could meet.} These priors represent the fact that agents may have idiosyncratic preferences---beauty being in the eye of the beholder---that are not captured by their type. We assume that these priors are common knowledge among the platform and the agents.

More formally, the utilities that two agents of types $\men\in\Men$ and $\women\in\Women$ derive if they match are drawn from some joint distribution $\Dist_{\men\women}$. Because we use $\Dist_{\men\women}$ to capture idiosyncrasies in preferences, we assume that utilities are drawn independently from $\Dist_{\men\women}$ for each such pair of agents.
To model mutual compatibility, we also assume agents have \emph{symmetric valuations}, meaning that if two agents match, then each agent derives the same utility $\utility\in\R$ from the match.

\vmedit{As mentioned in the introduction, an agent’s type captures their common features. 
In our abstract model, we assume that the platform directly observes these features and thus the agent's type.  While some features in practice may be self-reported by agents and thus present opportunities for  misreporting, we argue that this is not a major concern in our setting and motivating applications.
In practice, repeated meetings make verification easier: For example, if one misreports their height, then that would be revealed after meeting a potential partner, who could report back to the platform. At the same time, untruthful reporting  may result in severe punishments for an agent, such as being permanently banned from the platform.\footnote{\vmedit{For example, \href{https://hinge.co/terms}{Hinge's Terms of Service} require each user to not ``{\em Misrepresent your identity, age, current or previous positions, qualifications, or affiliations with a person or entity.}'' Violating these terms results in being banned from the platform.}} On a more abstract level, in many settings of interest, e.g., horizontally differentiated markets, it is unclear whether misreporting even helps an agent.}


Under the symmetric valuations assumption, each joint distribution $\Dist_{\men\women}$ can be thought of as a distribution over $\R$ with cumulative distribution function $\CDF_{\men\women}$. To avoid a technical discussion of tiebreaking\footnote{While all of our arguments work with discrete distributions, one has to take care to define tiebreaking properly. On the other hand, note that any non-continuous distribution can be approximated by a continuous distribution by adding a small amount of (bounded) noise.} we make the mild assumption that $\Dist_{\men\women}$ is a continuous distribution for all $(\men, \women)$. Then $\CDF_{\men\women}$ is absolutely continuous and $\CDF_{\men\women}(\thresh) = \Pr(\utility\le\thresh)$ for $\utility\sim\Dist_{\men\women}$. Finally, to make our notation for distributions symmetric, we define $\Dist_{\women\men}\coloneqq\Dist_{\men\women}$ and $\CDF_{\women\men}\coloneqq\CDF_{\men\women}$ for all $(\men, \women)$.

\subsection{Directed Search}\label{sec:directedsearch}

A key feature of our model is that search is mediated by a platform. At a high level, the platform directs search by picking for each agent type $\type$ the rates at which they ``meet'' other agent types while searching. For each agent, search consists of a sequence of meetings with potential partners until either a match is formed or they leave unmatched. When two agents meet, each must decide whether to accept or reject the match. If both parties accept (i.e., there is a double coincidence of wants), then they match with each other and leave the market. Otherwise, both agents remain in the market and continue their respective searches.

At the individual level, the frequency at which an agent encounters other agents of any given type follows a Poisson process with rate given by the platform's directed search policy. Upon two agents of types $\type$ and $\typeii$ meeting, their valuation $\utility$ for the match is drawn from $\Dist_{\type\typeii}$ and revealed to them. Each agent then decides whether to accept or reject the match, with the match forming if and only if both accept. If the agents match, each leaves the market with utility $\utility$; otherwise, each continues their search. We assume agents are rate-limited and normalize time so that agents can meet at most one candidate per unit time on average. That is, we require the total frequency of meetings for any agent to follow a Poisson process with at most unit rate---this is our \emph{capacity constraint}.

At the platform level, directed search can be thought of as some mass $\glow_{\men\women}$ of type $\men$ and type $\women$ agents, for each pair of types $(\men, \women)$, meeting each other per unit time. In particular, \emph{equal} masses of type $\men$ and type $\women$ agents should meet each other over any time interval. Note that this property must be satisfied if we were to explicitly pair agents in a discrete model. While our model is continuous, we can still imagine the platform ``pairing'' equal masses of randomly sampled agents for meetings at each instant of time. By requiring that agent-level Poisson processes be consistent with this fact, we obtain our \emph{flow balancing} constraint.

We now state the directed search model sketched above in more formal terms. For any $\type\in\Types$, define an \emph{assortment} to be a collection of rates $\rate_{\type}(\typeii)\ge 0$ for types $\typeii$ on the opposite side such that type $\type$ agents meet type $\typeii$ agents following a Poisson process with rate $\rate_{\type}(\typeii)$. The platform, as part of its design, chooses an assortment $\rate_\type$ for each type $\type\in\Types$. \vmdelete{(Note that we implicitly restrict the platform to designs that are anonymous and independent of agents' histories.)}

A set of assortments $\{\rate_\type\}_{\type\in\Types}$ is \emph{feasible} if it satisfies the following two sets of constraints:
\begin{description}
    \item[Capacity.] The total rate of meetings for type $\type$ agents is bounded by $1$ for all types $\type\in\Types$:
        \begin{equation}\label{eq:capacity}
            \sum_{\typeii} \rate_\type(\typeii)\le 1.
        \end{equation}
        (Note that $\sum_{\typeii} \rate_\type(\typeii)$ is the total rate because merging Poisson processes sums their rates.)
    \item[Flow balance.] For all pairs $(\men, \women)\in\Men\times\Women$, the mass of type $\women$ agents met by type $\men$ agents per unit time equals the mass of type $\men$ agents met by type $\women$ agents per unit time. By an exact law of large numbers \citep{sun2006elln, duffie2018dynamic}\footnote{As in the related literature on search and matching, we formally require a continuous-time exact law of large numbers for random matching; such a result has only recently been developed rigorously for discrete-time settings by \citet{duffie2018dynamic}.}, we can write this condition as
        \begin{equation}\label{eq:flowbalance}
        \glow_{\men\women}\coloneqq\mass_\men\rate_{\men}(\women) = \mass_\women\rate_{\women}(\men).
        \end{equation}
        That is, the mass $\glow_{\men\women}$ of type $\men$ and $\women$ agents that meet per unit time is well-defined.
\end{description}
We assume the platform can direct search according to any feasible set of assortments.

\awedit{To motivate the Poisson assumption, we think of the platform as implementing a feasible set of assortments by randomly sampling a subset of each type to meet some other type at each instant in time. With uniform sampling, individual agents will experience meetings at a Poisson rate. Note that by assuming uniform sampling, we implicitly restrict the platform to designs that are anonymous and independent of agents' histories.}

\section{The Platform's Design Problem}\label{sec:design}

In this section, we introduce the platform's design problem, namely: \emph{How should the platform suggest meetings in directed search to maximize social welfare in equilibrium?} In order to formally define the problem, we first describe the agents' strategic decision-making problem and briefly characterize the equilibria that arise.

\subsection{Strategies and Equilibrium Play}\label{sec:strategies}

Any set $\{\rate_\type\}_{\type\in\Types}$ of assortments defines a game for the agents, in which agents strategically decide which potential partners to accept and reject in their meetings. As our focus will be on stationary equilibria, we assume agents play \emph{time-} and \emph{history-independent} strategies. We further assume that strategy profiles are \emph{symmetric} within each type. We can thus write the strategy of type $\type$ agents as a function $\strategy_\type(\typeii, \utility)$ taking values in $[0, 1]$, specifying the probability with which a type $\type$ agent accepts when meeting a type $\typeii$ agent whom they value at utility $\utility$.

A property of our setup is that no individual agent can directly influence the behavior of other agents, as agents are only affected by the action profiles of types \emph{in aggregate}. Thus, rather than modeling agents as playing a full-fledged dynamic game, we can think of each agent as facing a Markov decision process (MDP) derived from their assortment $\rate_\type$ and strategies $\strategy_\typeii$ of the agents on the opposite side.\footnote{Note that agents are not directly affected by the actions of other agents on the same side.} Then, a strategy $\strategy_\type$ is a best response for type $\type$ agents if and only if it is an optimal policy for this MDP. It follows that a strategy profile $\{\strategy_\type\}_{\type\in\Types}$ is a Nash equilibrium if and only if the strategy of each type is an optimal policy for the MDP given by their assortment and opposite side's strategies.

In \Cref{sec:mdp}, we formally describe the agents' MDPs and show that any Nash equilibrium in ``non-dominated'' strategies is equivalent to one in \emph{threshold strategies}, i.e., strategies $\sigma_\type$ for which there exists a threshold $\thresh_\type$ such that $\sigma_\type(\typeii, \utility) = 1$ if $\utility\ge\thresh_\type$ and $\sigma_\type(\typeii, \utility) = 0$ if $\utility < \thresh_\type$.\footnote{In our analysis, we eliminate dominated strategies that reject matches worth more than one's expected continuation utility. Performing such a pruning is necessary to rule out degenerate equilibria (e.g., the one where agents reject all of their potential matches because they do not expect anyone to ever accept).} In fact, each threshold $\thresh_\type$ in this equilibrium equals the expected utility $\contutility_\type$ of type $\type$ agents upon entry. (Note that $\contutility_\type$ depends on type $\type$'s strategy, the other agents' strategies, as well as the platform's assortments---see \Cref{lemma:payoff}.) Furthermore, we show in \Cref{sec:formalities-unique} that
the equilibrium prediction of our model is unique in the following sense: For any set of assortments, our model makes a unique prediction of the distribution of matches that are realized. We record these observations in the following proposition \vmedit{(proven in \Cref{sec:formalities-unique})}:

\begin{proposition}[Structure of Nash Equilibria]\label{proposition:unique-main}
Given any fixed set of assortments $\{\rate_\type\}_{\type\in\Types}$:
\begin{enumerate}
    \item There exists a unique strategy profile where each type plays a threshold strategy with threshold equal to that type's expected utility. Moreover, this strategy profile is a Nash equilibrium.
    \item Any Nash equilibrium in non-dominated strategies produces the same distribution (in both types and utilities) of realized matches as the preceding Nash equilibrium in threshold strategies.
\end{enumerate}
\end{proposition}


In light of \Cref{proposition:unique-main}, it suffices to focus our attention on equilibria in threshold strategies. To facilitate our discussion of such equilibria, we state expressions for the rate $\xi_\type$ at which agents match \vmedit{(in \Cref{lemma:rate})} and their expected utility $\contutility_\type$ when all agents play threshold strategies \vmedit{(in \Cref{lemma:payoff})}.

\begin{lemma}[Matching Rate]\label{lemma:rate}
Suppose each type $\type\in\Types$ plays a threshold strategy with threshold $\thresh_\type$. Then the rate $\xi_\type$ at which individual type $\type$ agents match is
\begin{equation}\label{eq:rate}
\xi_\type = \sum_{\typeii} \paren*{\rate_\type(\typeii)\int_{\max(\thresh_\type, \thresh_\typeii)}^\infty \,d\CDF_{\type\typeii}}.
\end{equation}
\end{lemma}

\begin{proof}
The rate at which an agent of type $\type$ matches is the sum of the rates at which they meet each type $\typeii$ weighted by the probability such a meeting results in a match. The latter probability is \smash{$1 - \CDF_{\type\typeii}(\max(\thresh_\type, \thresh_\typeii)) = \int_{\max(\thresh_\type, \thresh_\typeii)}^\infty \,d\CDF_{\type\typeii}$} in terms of the agents' thresholds $\thresh_\type$ and $\thresh_\typeii$.
\end{proof}

\begin{lemma}[Expected Utility]\label{lemma:payoff}
Suppose each type $\type$ plays a threshold strategy with threshold $\thresh_\type$. Then the expected payoff\, $\contutility_\type$ of type $\type$ agents is
\begin{equation}\label{eq:payoff}
\contutility_\type = \frac{\sum_{\typeii}\paren[\Big]{\rate_{\type}(\typeii)\int_{\max(\thresh_\type, \thresh_\typeii)}^\infty\utility \,d\CDF_{\type\typeii}}}{\delta + \sum_{\typeii}\paren[\Big]{\rate_{\type}(\typeii)\int_{\max(\thresh_\type, \thresh_\typeii)}^\infty \,d\CDF_{\type\typeii}}}.
\end{equation}
(Note that by \Cref{lemma:rate}, the denominator can also be written as $\delta + \xi_\type$.)
\end{lemma}

While the proof of \Cref{lemma:payoff} depends on the formal definition of the MDP, the intuition for the formula is simple: For any type $\type$ agent, the memorylessness of the MDP implies their expected utility equals their expected utility conditioned on leaving at any particular instant. Conditioning on such an event, we see that with probability proportional to $\delta$, they left due to a life event, and with probability proportional to \smash{$\rate_\type(\typeii)\int_{\max(\thresh_\type, \thresh_\typeii)}^\infty \,\CDF_{\type\typeii}$}, they left because they matched with an agent of type $\typeii$. In equation \eqref{eq:payoff}, we simply express the agent's expected utility as the sum of their expected utilities in each of these scenarios weighted by probability of occurrence.

\subsection{Optimal Directed Search}\label{sec:optimaldirectedsearch}


The platform's design problem is to design the search process via assortments in order to induce a stationary equilibrium that is (approximately) optimal in terms of social welfare. Thus, in addition to designing the assortments, the platform should ensure that those assortments are sustainable in equilibrium. When analyzing the design problem, we assume that the platform has access to the \vmedit{departure} rate $\delta$, the arrival rates $\alpha_\type$ for all $\type\in\Types$, and the preference distributions $\Dist_{\men\women}$ for all $(\men, \women)\in\Men\times\Women$.

More formally, we say that a set of assortments $\{\rate_\type\}_{\type\in\Types}$ 
induces a \emph{stationary equilibrium} if, for the equilibrium strategy profile $\{\strategy_\type\}_{\type\in\Types}$, there exist population masses $\{\mass_\type\}_{\type\in\Types}$ such that:
\begin{enumerate}
    \item Stationarity (as defined in \cref{eq:stationarity}) holds for all types $\type\in\Types$.
    \item The assortments are feasible with respect to the population masses (i.e., \eqref{eq:capacity} and \eqref{eq:flowbalance} hold).
\end{enumerate}
By \Cref{proposition:unique-main},
a set of assortments $\{\rate_\type\}_{\type\in\Types}$ implies a unique Nash equilibrium in threshold strategies.
Furthermore, these threshold strategies determine $\xi_\type$ for each $\type$ by \Cref{lemma:rate}. Given $\{\xi_\type\}_{\type\in\Types}$, the stationarity condition then determines $\mass_\type$ for each $\type$. Therefore, any choice of assortments $\{\rate_\type\}_{\type\in\Types}$ induces at most one stationary equilibrium, making the platform's optimization problem over the assortments well-defined.\footnote{Since our model uniquely predicts play given assortments, stationary equilibria are also self-sustaining once established.}

Recall that the platform's optimization objective is social welfare. For stationary equilibria, social welfare is defined as the total flow rate of utility realized by agents matching with each other. 
Note that the symmetric valuations assumption implies each of the two sides contributes exactly half of the total welfare.
The platform's optimization problem is thus to find the assortments $\{\rate_\type\}_{\type\in\Types}$ that induce a stationary equilibrium with maximum welfare.

Our setup naturally lets us cast the platform's design problem of finding a welfare-maximizing policy for directed search as a computational one: Given inputs $\delta$, $\{\alpha_\type\}_{\type\in\Types}$, and $\{\Dist_{\men\women}\}_{(\men, \women)\in\Men\times\Women}$, \emph{compute} the assortments $\{\rate_\type\}_{\type\in\Types}$ that induce a stationary equilibrium with maximum welfare. We name this problem $\PlatformOpt$.\footnote{To formalize our computational problem while avoiding details relating  numerical computation, we assume that the platform's knowledge of the preference distributions $\Dist_{\men\women}$ comes in the form of oracle access to the CDFs $\CDF_{\men\women}$ and the tail expectations $\int_{\thresh}^\infty \utility \,d\CDF_{\men\women}$. We further assume that the platform can solve for each agent's best response optimally. This is a mild assumption because the corresponding optimization problem boils down to a ternary search on a unimodal function of the aforementioned oracle's outputs (see \Cref{lemma:rho} and its generalization \Cref{lemma:rho-2} in \Cref{sec:formalities}).}

We remark that we do not dwell on dynamical issues beyond steady state (e.g., convergence to equilibrium from ``cold starts''). Instead, we assume that the platform has the power to manipulate entry into the market at time $0$ and can ensure that the requisite populations of agents are present for the desired equilibrium to sustain indefinitely. This motivates our objective for the platform of finding a welfare-maximizing stationary equilibrium.

\subsubsection{Optimal Directed Search as an Optimization Problem}

Stationarity, feasibility, and agents playing best responses can all be encoded (albeit indirectly) as constraints on the assortments. We can thus write $\PlatformOpt$ as the following constrained optimization problem:
\begin{subequations}\label{eq:opt}
\begin{alignat}{5}
\max_{\rate_{\type}, \thresh_{\type}, \xi_\type, \mass_\type}&\quad&& 
\sum_{\men\in\Men}\sum_{\women\in\Women} \paren*{\mass_\men \rate_\men(\women)\int_{\mathrlap{\max(\thresh_\men, \thresh_\women)}}\qquad\utility\,d\CDF_{\men\women}} + 
\sum_{\women\in\Women} \sum_{\men\in\Men}&& \paren[\bigg]{\mass_\women\rate_\women(\men)\int_{\mathrlap{\max(\thresh_\men, \thresh_\women)}}\qquad\utility\,d\CDF_{\men\women}}\tag{\ref{eq:opt}} \\[0.5em]
\text{such that}%
&&& \alpha_\type = (\delta + \xi_\type)\mass_\type&&\forall\type\in\Types\label{eq:opt-stationarity} \\
&&& \mass_\men\rate_\men(\women) = \mass_\women\rate_\women(\men)&&\forall (\men, \women)\in\Men\times\Women\label{eq:opt-flowbalance} \\
&&& 1\ge \sum_{\typeii}\rate_\type(\typeii)&&\forall\type\in\Types\label{eq:opt-capacity} \\
&&& \xi_\type = \sum_{\typeii} \paren*{\rate_\type(\typeii)\int_{\max(\thresh_\type, \thresh_\typeii)}^\infty \,d\CDF_{\type\typeii}}&&\forall\type\in\Types\label{eq:opt-rate} \\
&&& \thresh_\type = \frac{1}{\delta + \xi_\type}\sum_{\typeii}\paren*{\rate_{\type}(\typeii)\int_{\max(\thresh_\type, \thresh_\typeii)}^\infty\utility \,d\CDF_{\type\typeii}}&&\forall\type\in\Types\label{eq:opt-payoff} \\
&&& \rate_\men(\women), \rate_\women(\men)\ge 0 && \forall (\men, \women)\in\Men\times\Women.\label{eq:opt-nonnegativity}
\end{alignat}
\end{subequations}
Here, constraint \eqref{eq:opt-stationarity} is the stationarity constraint \eqref{eq:stationarity}; constraints \eqref{eq:opt-flowbalance} and \eqref{eq:opt-capacity} are the feasibility constraints of flow balance \eqref{eq:flowbalance} and capacity \eqref{eq:capacity}; constraint \eqref{eq:opt-rate} is formula \eqref{eq:rate} for individuals' rate of matching; constraint \eqref{eq:opt-payoff} states that each type's threshold is a fixed point of their MDP payoff function \eqref{eq:payoff}; and constraint \eqref{eq:opt-nonnegativity} simply requires that the meetings rates be non-negative.

The validity of constraint \eqref{eq:opt-payoff} follows from \Cref{proposition:unique-main}: Our model's unique prediction for which matches are realized occurs when agents all play threshold strategies with threshold equal to their expected utility. Moreover, the profile of thresholds satisfying this property is uniquely determined by the assortments. Hence the social welfare of an assortment is given by the flow of utility realized when all agents play such threshold strategies.

We note that while the optimization program \eqref{eq:opt} is written as an optimization problem over four sets of variables $\{\rate_\type\}_{\type\in\Types}$, $\{\thresh_\type\}_{\type\in\Types}$, $\{\xi_\type\}_{\type\in\Types}$, and $\{\mass_\type\}_{\type\in\Types}$, it is really an optimization problem over the assortments $\{\rate_\type\}_{\type\in\Types}$: Agents' thresholds $\{\thresh_\type\}_{\type\in\Types}$ and matching rates $\{\xi_\type\}_{\type\in\Types}$ are determined entirely by their assortments, and the stationary masses $\{\mass_\type\}_{\type\in\Types}$ are in turn determined entirely by the matching rates.


\subsection{Illustrative Examples: Directed Search in Horizontal and Vertical Markets}\label{sec:examples}

With our setup of the model and the platform's design problem now complete, we present 
\vmedit{two simple examples with fundamentally different market structures: a horizontal market and a vertical one. In these special cases, we are able to optimally solve $\PlatformOpt$ and thus gain insight into the structure and value of the optimal design. } 
\vmedit{Through the former example, we show that platform-directed search can provide significant welfare gains relative to random meeting. Through the latter, we illustrate some subtleties involving congestion and cannibalization that arise (even in this very restricted setting) when solving $\PlatformOpt$.}

\vmdelete{as an example a market where agents have deterministic and strictly vertical preferences. Through this example, we illustrate some subtleties involving congestion and cannibalization that arise (even in this very restricted setting) when solving $\PlatformOpt$.}


To motivate the utility of platform-directed search, we describe a market involving highly heterogeneous preferences (i.e., horizontal differentiation) where directed search leads to significantly greater efficiency than random meeting.

\begin{example}[Horizontal market]
\label{example:horizontal}
We consider a market where the two sides are $\Men = \{\men_1,\ldots,\men_\ntypes\}$ and $\Women = \{\women_1,\ldots,\women_\ntypes\}$, with types $\men_i$ and $\women_j$ achieving positive utility from matching only if $i = j$. Suppose also that the types are symmetric, with each type having the same arrival rate $\alpha_\type = 1$ and each pair of compatible types $(\men_i, \women_i)$ having the same preference distribution $\CDF_{\men_i\women_i} = \CDF$. Finally, let $\delta > 0$ be arbitrary.

We compare the welfares of two stationary equilibria---one under random meeting and another under optimal directed search. Consider the symmetric equilibrium where agents meet randomly. By symmetry, only a $1/n$ fraction of each agent's meetings are with agents of the compatible type. On the other hand, with the optimal directed search, the platform sets assortments so that $\rate_{\men_i}(\women_i) = \rate_{\women_i}(\men_i) = 1$ for all $i$.\footnote{This policy is optimal because its welfare matches the that of the first-best outcome (defined in \Cref{sec:reduction}).} Then, agents meet other agents of the same type with probability $1$. To compare these two equilibria, observe that the random meeting equilibrium is equivalent to the optimal directed search equilibrium with search friction $n\delta$---under random meeting, each agent has a $n$ times higher probability of leaving unmatched between meetings with compatible agents. In the limit of large $\delta$, where agents are likely to accept on their first meeting (or not match at all), this leads to a multiplicative gap approaching $n$ between the welfares of random meeting and optimal directed search.
\end{example}

This example demonstrates that in a market with horizontal preferences, there can be a significant amount of wasteful meeting if agents proceed with random search: In our example, under random meeting, only one in $\ntypes$ meetings involves pairs that are compatible with each other. On the other hand, with directed search, the platform could eliminate such wasteful meeting entirely. While admittedly stylized, this example nonetheless highlights the usefulness of directed search in horizontally differentiated markets.

\vmdelete{With our setup of the model and the platform's design problem now complete, we present as an example a market where agents have deterministic and strictly vertical preferences. Through this example, we illustrate some subtleties involving congestion and cannibalization that arise (even in this very restricted setting) when solving $\PlatformOpt$.}

In our \vmedit{second example, of a} vertically structured market, agents on each side of the market have a common preference ordering over the types on the opposite side. In such a market, one might expect that the optimal directed search scheme results in a positively assortative matching when agents' valuations are supermodular. However, this turns out not to always be the case. Because agents arrive at disparate rates, matching assortatively can lead to inefficiency due to congestion at the top of the market, and due to ``high type'' agents cannibalizing demand for ``low type'' agents. While such phenomena make finding the exact optimal assortments challenging, we will show in \Cref{sec:approximation} that we can nonetheless find a set of assortments that are approximately optimal.

\begin{example}[A Vertical Market]
\label{example:vertical}
Consider a market with strictly vertical preferences, where the agents on each side belong to either a high type $H$ or a low type $L$. That is, we let $\Men = \{\men_H, \men_L\}$ and $\Women = \{\women_H, \women_L\}$. Suppose the utility distributions $\CDF_{\men_H\women_H}$, $\CDF_{\men_L\women_H}$, $\CDF_{\men_H\women_L}$, and $\CDF_{\men_L\women_L}$ are point masses at $\utility(1 + 2\exeps)$, $\utility$, $1 + \exeps$, and $1$, respectively, for $\utility\gg 1$ and $\exeps\ll 1$.\footnote{An astute reader may notice that, technically, our preference distributions do not satisfy our continuity assumption. We present our example this way for simplicity's sake---the same example can be made to work with continuous distributions by adding some small amount of Gaussian noise to each utility value.} Thus, agents always prefer matching with high-type agents over low-type agents. Suppose further that the arrival rates are such that $\alpha_{\men_H} = \alpha_{\women_L} = 1$ and $\alpha_{\women_H} = \alpha_{\men_L} = 1 / \utility$. Finally, we consider this example in the regime where $\delta$ is very small (i.e., suppose we are in a nearly frictionless market).

To analyze the possible stationary equilibria of this market, we first consider what happens in the equilibria where high-types only match with high-types and low-types only match with low-types. The maximum possible welfare in such a equilibrium is bounded above by $2((1 + 2\exeps) + \alpha_{\men_L}) / (1 + \delta)\approx 2$, since the welfare generated by the high-types is at most $2\cdot\alpha_{\women_H}\cdot\utility(1 + 2\exeps) / (1 + \delta)$ and the utility generated by the low-types is nearly negligible at $2\cdot\alpha_{\men_L}\cdot 1/(1 + \delta)$. In particular, matching for both high- and low-types is bottlenecked by the low arrival rate of the less common type. One can think of these equilibria as being congested at both levels of the market: Due to the disparities in arrival rates, there are many type $\men_H$ agents hoping to match with only a few type $\women_H$ agents and many type $\women_L$ agents hoping to match with only a few type $\men_L$ agents. As a result, large numbers of type $\men_H$ agents and type $\women_L$ agents go unmatched.

Next we observe that such equilibria---where high-types only match with high-types and low-types only match with low-types---arise if the expected utility $\contutility_{\men_H}$ of type $\men_H$ is greater than $1 + \exeps$.
Notably, type $\men_H$ agents \emph{are not willing} to match with type $\women_L$ agents due to the possibility of matching with a type $\women_H$ agent: If $\contutility_{\men_H} > 1 + \exeps$, then \Cref{proposition:unique-main} tells us that type $\men_H$ agents will not match with type $\women_L$ agents, since type $\men_H$ agents would only obtain $1 + \exeps$ utility from doing so.
From the platform's perspective, the presence of $\women_H$ agents in type $\men_H$'s assortment cannibalizes demand for type $\women_L$ agents.
Likewise, the possibility of matching with type $\men_H$ makes type $\women_H$ reject type $\men_L$ agents---the expected utility of type $\women_H$ agents is at least $(1 + \exeps)\cdot\alpha_{\men_H} / \alpha_{\women_H} > \utility$ from matching with type $\men_H$.


On the other hand, we can obtain a more efficient equilibrium if the expected utility of type $\men_H$ is slightly lower, e.g., if they only match with type $\women_L$. Consider the assortment where type $\men_H$ and $\women_L$ agents only meet each other and type $\women_H$ agents and $\men_L$ agents only meet each other. All agents meet at maximal rates; furthermore, all meetings result in matches because no type has other options. As a result, the social welfare in equilibrium is $2(\alpha_{\men_H}\cdot(1 + \exeps) + \alpha_{\men_L}\cdot\utility) / (1 + \delta) = 2((1 + \exeps) + 1) / (1 + \delta)\approx 4$. That is, the social welfare of this non-assortative ``diagonal'' set of assortments (which happens to be optimal) is nearly twice that of any equilibrium where high-types only match with high-types and low-types only match with low-types.
\end{example}

\vmedit{We conclude our discussion of the above market with two remarks. First, the cannibalization effect (and its welfare implications) is in the same spirit as the {\em sorting externality} discussed in  \citet{BurdettColes97} using a high/low type vertical market, but with a different preference structure. In that setting, the preference of an agent only depends on the type of the other side, and agents meet each other uniformly at random. Depending on the belief that type $H$ would match with a type $L$ or not (and its corresponding match threshold), two different equilibria arise. Interestingly, the assortative equilibrium has the lower welfare  compared to the equilibrium under which type $H$ matches both types. Next, we}
note these phenomena arise in part due to the fact that the platform 
\vmedit{design is anonymous and thus it}
cannot discriminate between agents beyond their type. Indeed, at the opposite extreme, the platform could simply force agents to match with a certain type, e.g., by having a subset of type $\men_H$ agents only meet type $\women_L$ agents and another subset of type $\men_H$ agents only meet type $\women_H$ agents. However, 
\vmedit{we stress that such splitting of types has two key drawbacks: (i) if agents are able to rejoin the platform under a different identity, then such separation is difficult to enforce; (ii) such discriminatory treatment may be unpalatable to users from a fairness perspective.}
We also show in \Cref{sec:approximation} that such an assumption is not restrictive---our solution is constant-factor competitive against the \emph{first-best} platform policy.


%
%
%
%
\section{Approximation Algorithms}\label{sec:approximation}

In this section, we show how the platform can construct in polynomial time a set of assortments $\{\rate_\type\}_{\type\in\Types}$ such that its induced equilibrium obtains a $4$-approximation to the optimal social welfare. Furthermore, the equilibrium resulting from our computed assortments has a simple and appealing structure: It consists of disjoint submarkets such that each submarket has only a single type on one of the sides. 
We also establish hardness of approximation for the platform's optimization problem, meaning that we cannot do better than a constant-factor approximation unless $\mathsf{P} = \NP$.

We actually show a slightly stronger result: The resulting equilibrium obtains a $4$-approximation to the optimal welfare that could be obtained in the first-best matching economy, where rather than agents choosing to accept or reject, the platform gets to decide on behalf of the agents. (In this first-best economy, the platform gets to choose the assortments as well.) Our theorem can thus also be interpreted as a price-of-anarchy-style result, that agents' self-interested behavior causes welfare to degrade by at most a factor of $4$ relative to the first-best optimum.


\begin{theorem}[Computationally Efficient Approximation]\label{theorem:approximation}
There is a polynomial-time $4$-approximation algorithm for $\PlatformOpt$. (In fact, this algorithm gets a $4$-approximation to the platform's first-best optimal welfare.)
\end{theorem}

Our construction of these assortments and the resulting equilibrium proceeds as follows:
\begin{description}
    \item[Solving for the platform's first-best.] We start by considering the platform's optimization problem for the first-best matching economy, where the platform gets to plan both the assortments and which matches the agents accept. That is, we relax the agents' incentive constraint and allow the platform to optimize over thresholds $\thresh_{\men\women}$ between each pair $(\men, \women)$ of types. This relaxation removes the main barrier---namely, incentives---to solving the optimization problem exactly. Accordingly, we show that this relaxed problem can be reduced to a generalization of the classical assignment problem and thus be efficiently solved.
    
    \item[Approximating via ``star-shaped'' markets.] We then show that an optimal solution to this ``generalized'' assignment problem can be converted into a stationary equilibrium (where we now take agents' incentives back into account) with at most a factor $4$ loss in welfare. The key ideas for this step are:
    \begin{itemize}
        \item Showing that the optimal first-best solution can be $2$-approximated by another consisting of ``star-shaped'' submarkets (which have a side where all the agents are of the same type). For this step, we draw on some classical observations about the structure of optimal solutions from previous works on generalized assignment problems \citep{LenstraShmoysTardos90, ChakrabartyGoel10, banerjee2017segmenting}.
        \item Showing that for markets where all the agents are of the same type on one side, the platform can obtain a $2$-approximation to the optimal welfare of the first-best matching economy. This step relies on the observation that in such markets, there is an alignment of incentives between individual agents and the platform.
    \end{itemize}
\end{description}

A constant-factor approximation algorithm is 
likely 
the best we can hope for, computationally:  We show in \Cref{sec:hardness} that approximating the platform's optimization problem to a factor better than $\frac{24}{23}$ is $\NP$-hard.

\begin{theorem}[Hardness of Approximation]\label{theorem:hardness-main}
 Approximating $\PlatformOpt$ up to a $\paren*{\frac{24}{23} - \approxerr}$-factor is $\NP$-hard for any $\approxerr > 0$.
\end{theorem}

We prove the above hardness result in \Cref{sec:hardness}.
From a technical perspective, our approach very closely mirrors that of \citet{ChakrabartyGoel10}, who show hardness of approximation for related problems (e.g., maximum budgeted allocation) by reducing from Hastad's $3$-bit PCP \citep{Hastad01}. While their technique applies to our setting, it is not clear that their hardness results apply directly---a challenge specific to our setting is that our optimization problem is continuous rather than discrete. To make this approach work, we must show that our complicated feasibility set introduces a discrete element to the optimal allocation. But as a consequence, we obtain a slightly worse constant $c$.

A constant-factor gap between first-best and second-best is also inherent to the $\PlatformOpt$ problem. We show via example in 
\vmedit{\Cref{app:proof:first-second-gap}}
 that there exists a market such that the first-best and second-best are separated by a factor of $2$.

\begin{proposition}[Gap Between First- and Second-Best]\label{proposition:first-second-gap}
For any $\approxerr > 0$, there exists a market such that the platform's first-best optimal welfare is at least $2-\approxerr$ times the optimal  (second-best) welfare for $\PlatformOpt$.
\end{proposition}

\vmedit{Before proceeding with the proof of \Cref{theorem:approximation},  we remark that the first-best 
optimal welfare and the welfare achieved by random meeting serve as natural benchmarks to compare our solution against. (We remind the reader that, in our general setting, an assortative matching is not well-defined.) \Cref{theorem:approximation} ensures that our solution achieves at least $1/4$ of the former benchmark. On the other hand, \Cref{example:horizontal} implies the welfare achieved under our solution can be arbitrarily larger than that achieved by random meeting. In fact, it is straightforward to show that for the horizontal market of Example \ref{example:horizontal}, our solution coincides with the first-best.\footnote{\vmedit{See Section \ref{sec:simulation} for a numerical solution of a similar horizontal market as well as  further  numerical analysis.}} Consequently, for large $\delta$, its welfare gain compared to random meeting would be proportional to $n$.}

In the remainder of this section, we develop the proof of \Cref{theorem:approximation}.

%
%
%
%
%
%
%

\subsection{The First-Best Matching Economy}\label{sec:reduction}

In this section, we study the platform's first-best optimization problem, where the platform has the power to plan the entire matching economy. That is, in addition to designing assortments, we allow the platform to ignore agents' incentives and dictate each type's strategy. The first-best relaxation is obtained by modifying the optimization problem \eqref{eq:opt} as follows: Relax the best response constraint \eqref{eq:opt-payoff}. Then, replace the per-type threshold variables $\thresh_\type$ with pairwise threshold variables $\thresh_{\men\women}$ that take the place of $\max(\thresh_\men, \thresh_\women)$ in all expressions for all $(\men, \women)\in\Men\times\Women$ in \eqref{eq:opt}.

To see why this relaxation corresponds to the platform's first-best problem, note that rather than optimizing over all agent strategies, the platform can restrict its attention to strategies $\strategy_\type$ that have a threshold for each type on the opposite side. This is because whenever two agents of types $\men$ and $\women$ match at a certain rate, it is optimal for the matches realized to be the highest-valued ones. In particular, the platform should never hope that agents of types $\men$ and $\women$ reject each other at utility $\utility$ but match at some utility $\utilityii < \utility$. Rather, there should exist a threshold $\thresh_{\men\women}$ such that agents of these types match if and only if $\utility\ge\thresh_{\men\women}$.

While even the relaxed optimization problem looks difficult to work with, our goal in this section is to show that a careful reparametrization reduces the optimization problem to a linear program. The resulting linear program is related to (and in fact generalizes) the linear programming relaxation of the assignment problem; the main difference is that instead of having unit capacities on the flow to each node, we allow for real-valued capacities. Formally, 
the reduction can be stated as follows:

\begin{proposition}[First-Best ``Assignment Problem'']\label{proposition:lp}
    Consider the linear program
    \begin{subequations}\label{eq:lp}
    \begin{alignat}{3}
        \max_{\beta_{\men\women}} &\quad&& 2\cdot\sum_{\men\in\Men}\sum_{\women\in\Women} \rho_{\men\women}\cdot\beta_{\men\women} \tag{\ref{eq:lp}} \\
        \text{such that}&&& \sum_{\women\in\Women} \beta_{\men\women} \le\alpha_\men\quad\forall\men\in\Men\label{eq:lp-1} \\
        &&& \sum_{\men\in\Men} \beta_{\men\women} \le\alpha_\women\quad\forall\women\in\Women, \label{eq:lp-2}
    \end{alignat}
    \end{subequations}
    where
    \smash{$\rho_{\men\women} \coloneqq \max_{\thresh\ge 0} \frac{\int_{\thresh}^\infty u\,d\CDF_{\men\women}}{\delta + \int_{\thresh}^\infty \,d\CDF_{\men\women}}$}.
    The platform's first-best optimization problem is equivalent to \eqref{eq:lp} in the following sense:
    \begin{enumerate}
        \item The two optimization problems have the same optimal objective value.
        \item Any feasible choice of $\{\beta_{\men\women}\}_{\men\in\Men,\women\in\Women}$ corresponds to a feasible choice of 
        $\{\lambda_\type\}_{\type\in\Types}$, $\{\mass_\type\}_{\type\in\Types}$, $\{\xi_\type\}_{\type\in\Types}$, and $\{\thresh_{\men\women}\}_{\men\in\Men,\women\in\Women}$
        that achieves the same objective value.
    \end{enumerate}
\end{proposition}

\Cref{proposition:lp} implies that to solve the first-best optimization problem, it suffices to solve the linear program \eqref{eq:lp} and convert the resulting solution to an optimal choice of variables for the first-best optimization problem.  
%
The proof of the \Cref{proposition:lp} proceeds by first reparameterizing the optimization program \eqref{eq:opt} using the flow balance constraint: We introduce the variable
\[ \glow_{\men\women}\coloneqq\mass_\men\rate_{\men}(\women) = \mass_\women\rate_{\women}(\men) \]
for the total rate of meetings between each pair $(\men, \women)\in\Men\times\Women$.  
We can then encode the flow balance, capacity, and stationary constraints into the single combined constraint
\begin{equation}\label{eq:opt-combined-main}
    \alpha_\type\ge\sum_{\typeii}\paren*{\glow_{\type\typeii}\paren*{\delta + \int_{\thresh_{\type\typeii}}^\infty \,d\CDF_{\type\typeii}}}.
\end{equation}
That is, there exist corresponding values for $\rate_{\type}(\typeii)$, $\xi_\type$, and $\mass_\type$ that satisfy the constraints of \eqref{eq:opt} if and only if the $\glow_{\type\typeii}$ variables satisfy \eqref{eq:opt-combined-main}. To finish, we introduce the variables
\begin{equation}\label{eq:beta-main}
\beta_{\men\women}\coloneqq\glow_{\men\women}\paren*{\delta + \int_{\thresh_{\men\women}}^\infty \,d\CDF_{\men\women}}.
\end{equation}
We then substitute into \eqref{eq:opt-combined-main} and the objective of \eqref{eq:opt}. The latter becomes
\begin{equation}
2\sum_{\men\in\Men}\sum_{\women\in\Women}  \paren*{\beta_{\men\women}\cdot\frac{\int_{\thresh_{\men\women}}^\infty\utility\,d\CDF_{\men\women}}{\delta + \int_{\thresh_{\men\women}}^\infty \,d\CDF_{\men\women}}}.
\end{equation}
While the thresholds $\thresh_{\men\women}$ are not yet fixed, these substitutions make the optimization effectively unconstrained in $\thresh_{\men\women}$. Optimizing over $\thresh_{\men\women}$ in the objective lets us write the objective in terms of $\rho_{\men\women}$ as in \eqref{eq:lp}, since the only dependence on $\thresh_{\men\women}$ is through the expression
\begin{equation}\label{eq:rho-main}
\frac{\int_{\thresh_{\men\women}}^\infty\utility\,d\CDF_{\men\women}}{\delta + \int_{\thresh_{\men\women}}^\infty \,d\CDF_{\men\women}}
\end{equation}
We defer full details of this analysis to  \vmedit{\Cref{sec:proof:lp}}.

The final step in the above is setting $\thresh_{\men\women}$ so that expression \eqref{eq:rho-main} is maximized. It turns out that we can do so by setting $\thresh_{\men\women} = \rho_{\men\women}$. (That is, the expression achieves its maximum value at a fixed point.) We record this observation here: 

\begin{lemma}[Fixed Point Structure of Optimal Threshold]\label{lemma:rho}
    Define 
    \[ A(\thresh) = \frac{\int_\thresh^\infty \utility\,d\CDF}{\delta\, + \int_\thresh^\infty \,d\CDF}, \]
    where $\delta > 0$ and $\CDF$ is a continuous distribution, and let $\rho = \max_{\thresh\ge 0} A(\thresh)$. Then $A(\rho) = \rho$. Moreover, $A$ is monotonically increasing for $\thresh\le\rho$ and monotonically decreasing for $\thresh\ge\rho$.
\end{lemma}

\subsection{Designing an Approximately Optimal Equilibrium}\label{sec:starshaped}

In the previous section, we described how the platform can solve for a first-best solution, where agents' incentives are ignored. In this section, we now bring back agents' incentives and show how to convert the first-best solution into a stationary equilibrium. We show that our conversion process retains at least $\frac 14$ of the welfare of the first-best solution, thereby giving us an efficient $4$-approximation algorithm for $\PlatformOpt$.

\ifx\acmConference\undefined
    \begin{figure}
        \centering
        \sffamily
        \FIGURE
        {\includesvg[scale=0.42]{figures/outline.svg}}
        {This figure depicts the process for converting an optimal first-best solution to a stationary equilibrium, while preserving welfare up to a factor of $4$. We start with a first-best solution whose graph of positive weight variables forms a forest (i.e., step (1)). The first arrow represents step (2), where we approximate the forest with stars, each corresponding to a star-shaped submarket. The second arrow represents solving the first-best for each star-shaped submarket (see \Cref{lemma:first-best}). The third pair of arrows corresponds to the cases we handle in the proof of \Cref{proposition:starshaped}---in particular, whether the new first-best solution induces a stationary equilibrium (Case 1) or not (Case 2).\label{fig:outline}}{}
    \end{figure}
\else
    \begin{figure}
        \centering
        \sffamily
        \small
        \includesvg[scale=0.35]{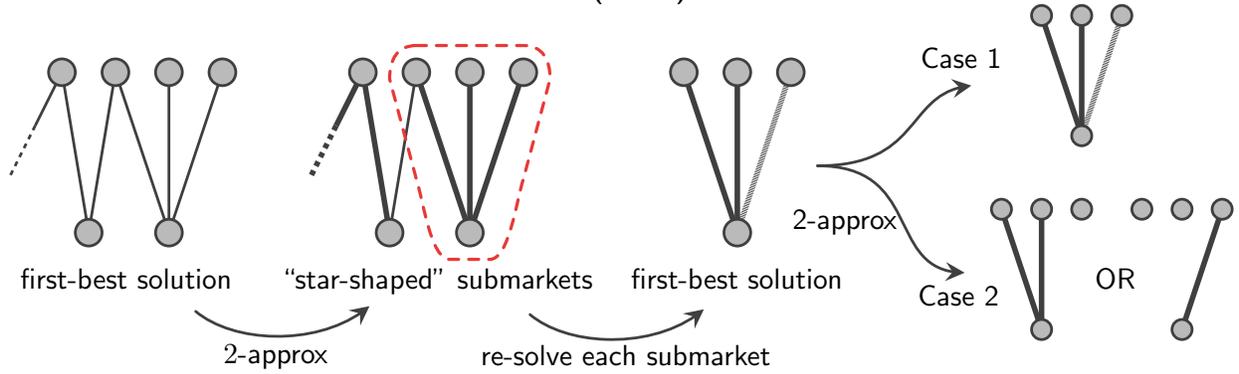}
        \caption{This figure depicts the process for converting an optimal first-best solution to a stationary equilibrium, while preserving welfare up to a factor of $4$. We start with a first-best solution whose graph of positive weight variables forms a forest (i.e., step (1)). The first arrow represents step (2), where we approximate the forest with stars, each corresponding to a star-shaped submarket. The second arrow represents solving the first-best for each star-shaped submarket (see \Cref{lemma:first-best}). The third pair of arrows corresponds to the cases we handle in the proof of \Cref{proposition:starshaped}---in particular, whether the new first-best solution induces a stationary equilibrium (Case 1) or not (Case 2).}
        \label{fig:outline}
    \end{figure}
\fi

%

%

At a high-level, our conversion process works as follows (see \Cref{fig:outline} for a visual depiction):
\begin{description}
  \item[Step (1).]
    Find a solution $\{\beta_{\men\women}^*\}_{\men\in\Men,\women\in\Women}$ to the linear program \eqref{eq:lp} such that the graph on $\Theta$ that has an edge $(\men, \women)$ for each $\beta_{\men\women}^* > 0$ is a forest.
  \item[Step (2).]
    Approximate this edge-weighted forest with a union of vertex-disjoint star graphs such that the star graphs contain at least half the total edge weight of the forest.
  \item[Step (3).]
    Construct a set of assortments for each star-shaped submarket (where one side has only one type of agent) that induces a stationary equilibrium whose social welfare is at least half the first-best social welfare of the submarket.
\end{description}
The first two of the above steps have appeared in other works featuring generalized assignment problems \cite{LenstraShmoysTardos90, ChakrabartyGoel10, banerjee2017segmenting}, where they were similarly used to simplify the problem structure. Our main insight for this stage of the proof lies in the third step, where we leverage structural properties of the agents' decision problem to show that the optimal stationary equilibrium in each submarket produces at least half the first-best welfare of that submarket.\footnote{Note that this last step is distinct from simply computing the optimal stationary equilibrium. Computing the equilibrium is not enough to show that it actually obtains at least half the first-best welfare.
}

These steps, together with \Cref{proposition:lp}, prove our approximation result \Cref{theorem:approximation}.

\subsubsection{Steps (1) and (2)}

We now state the lemmas that give the results needed for steps (1) and (2). The first lemma, having appeared in various forms since \citet{LenstraShmoysTardos90}, shows that there exists an optimal solution for our generalized assignment problem such that the positive variables $\beta_{\men\women} > 0$ form a forest when viewed as edges on a graph with vertex set $\Types$. This result is true for the same reason that, in the classical assignment linear program, there exists an optimal solution where all the variables are integral: As long as a cycle exists, one can push ``flow'' through this cycle to eliminate it while simultaneously reducing the value of some variable to $0$.

\begin{lemma}[\citet{LenstraShmoysTardos90, ChakrabartyGoel10, banerjee2017segmenting}]\label{lemma:forest}
  For a generalized assignment problem
  \begin{alignat*}{3}
    \max &\quad&& \sum_{\men\in\Men}\sum_{\women\in\Women} \rho_{\men\women}\cdot\beta_{\men\women} \\
    \text{such that}&&& \sum_{\women\in\Women} \beta_{\men\women} \le\alpha_\men\quad\forall\men\in\Men \\
                    &&& \sum_{\men\in\Men} \beta_{\men\women} \le\alpha_\women\quad\forall\women\in\Women,
  \end{alignat*}
  with real-valued capacities $\alpha_\type$, there exists an optimal solution $\{\beta_{\men\women}^*\}_{\men\in\Men,\women\in\Women}$ such that the graph on $\Types$ with edge set $\{ (\men, \women) : \beta_{\men\women}^* > 0 \}$ is a forest.
\end{lemma}



The next lemma, which is the main ingredient for step (2), shows that we can take any edge-weighted forest and approximate it up to a factor of $2$ in terms of edge weight with a disjoint union of star graphs. The precise formulation we cite is due to \citet{banerjee2017segmenting}, though simillar ideas have also appeared in the contexts of budgeted allocations \cite{LenstraShmoysTardos90, ChakrabartyGoel10}. This lemma is simple to prove: Root each tree in the forest arbitrarily, and color each edge red or blue based on its distance modulo $2$ from its root. The colors partition the forest into two subgraphs each made up of star-shaped graphs; one of these subgraphs must capture at least half the edge weight of the forest.\footnote{We remark that the constant factor on this lemma is tight: Consider the tree with $2\nfill+1$ vertices, where the root vertex has $\nfill$ children and each of these children has another child. Then, any subgraph that is a vertex-disjoint union of star graphs can only obtain edge weight $\nfill+1$.}

\begin{lemma}[\citet{banerjee2017segmenting}]\label{lemma:star}
  Given an edge-weighted forest, there exists a subgraph that is a union of vertex-disjoint star graphs (i.e., trees of radius $1$) such that the total edge weight of the subgraph is at least half that of the original tree.
\end{lemma}


In context, \Cref{lemma:forest} tells us that we can find a solution $\{\beta_{\men\women}^*\}_{\men\in\Men,\women\in\Women}$ to the linear program \eqref{eq:lp} for the platform's first-best optimization problem such that the positive weight variables form a forest on $\Types$. We weight the edges of this forest by their contribution $\beta_{\men\women}\cdot\rho_{\men\women}$ to social welfare in the first-best optimization problem. Next, \Cref{lemma:star} tells us that we can $2$-approximate this edge-weighted forest with a union of vertex-disjoint star graphs. Each of these star graphs defines a ``star-shaped'' submarket; note that each such submarket has a side where all agents are of the same type. This completes steps (1) and (2).

\subsubsection{Step (3)}\label{sec:step-3}

We now show that for any star-shaped submarket (i.e., one where one side of the market consists entirely of agents of the same type), the platform can design a set of assortments that induce a stationary equilibrium with welfare at least half that of the platform's \emph{first-best} solution. Formally, we prove the following proposition:

\begin{proposition}[Star-shaped Markets]\label{proposition:starshaped}
  For any star-shaped market, the platform can choose assortments that induce a stationary equilibrium whose welfare is a $2$-approximation to the platform's first-best welfare.
\end{proposition}

For the proof of this proposition, a key observation is that the platform's first-best threshold coincides with an agent's optimal threshold if their assortment is saturated with agents of a single type. In particular, such agents will not accept more than the platform would want them to in the first-best equilibrium. We obtain from this observation a condition for when the first-best solution induces a stationary equilibrium. We then show that the induced equilibrium achieves at least half the first-best welfare. When the condition fails to hold, it turns out that we can modify the first-best solution appropriately, by only matching agents whose capacity constraint is tight, and still obtain an equilibrium where at least half the first-best welfare is preserved.

For the remainder of this section, we assume without loss of generality that $\Men = \{\men\}$, i.e., our star-shaped market is such that the $\Men$ side of the market consists only of type $\men$ agents.
    
\paragraph{Structural Observations for Star-Shaped Markets.} Before giving the proof of \Cref{proposition:starshaped}, we develop some structural observations about the optimal first-best solution and agents' thresholds. We defer the proofs of all lemmas in this section to \Cref{sec:approximation-proofs-step-3}.

In our first lemma, we characterize the platform's first-best solution in star-shaped markets:

\begin{lemma}[First-best Solution of a Star-Shaped Market]\label{lemma:first-best}
There exists a first-best solution $\{\beta_{\men\women}^*\}_{\women\in\Women}$ and a subset $\women_1,\women_2,\ldots,\women_\nfill\subseteq\Women$ such that $\rho_{\men\women_{1}}\ge\rho_{\men\women_{2}}\ge\cdots\ge\rho_{\men\women_\nfill}$, $\beta_{\men\women_i}^* = \alpha_{\women_i}$ for all $i < \nfill$, and $\beta_{\men\women}^* = 0$ if $\women\not\in\{\women_1,\ldots,\women_\nfill\}$.
\end{lemma}

Now that we understand the platform's first-best optimal solution, we can start to incorporate incentives through the incentive constraint \eqref{eq:opt-payoff}.
By \Cref{lemma:rho}, each first-best threshold $\thresh_{\men\women_i}$ in the first-best solution described in \Cref{lemma:first-best} can be taken to be $\rho_{\men\women_i}$, since $\rho_{\men\women_i}$ maximizes
\[ \frac{\int_{\thresh_{\men\women_i}}^\infty \utility\,d\CDF_{\men\women_i}}{\delta + \int_{\thresh_{\men\women_i}}^\infty\,d\CDF_{\men\women_i}}. \]
Then, by the definition \smash{$\beta_{\men\women_i} = \glow_{\men\women_i}\paren*{\delta + \int_{\thresh_{\men\women_i}}^\infty \,d\CDF_{\men\women_i}}$} in \eqref{eq:beta-main}, this first-best solution corresponds to the choice of $\glow_{\men\women_i}$ given by
\[ \glow_{\men\women_i} = \glow_{\men\women_i}^* \coloneqq \frac{\beta_{\men\women_i}^*}{\delta + \int_{\rho_{\men\women_i}}^\infty \,d\CDF_{\men\women_i}}. \]

Our next step is to consider incentives and characterize equilibrium when $\glow_{\men\women_i} = \glow_{\men\women_i}^*$. Fixing the $\glow_{\men\women_i}$ values gives us information about equilibrium thresholds $\thresh_{\type}$ for all types because the incentive constraint \eqref{eq:opt-payoff} is equivalent to
\begin{equation}\label{eq:fixedpoint}
\thresh_\type = \frac{1}{\alpha_\type}\sum_{\typeii}\paren*{\glow_{\type\typeii} \int_{\max(\thresh_\type, \thresh_\typeii)}^\infty \utility\,d\CDF_{\type\typeii}}
\end{equation}
after scaling both the numerator and the denominator of the right-hand side by $\mass_\type$. (Recall that this fixed point definition of equilibrium play comes from \Cref{proposition:unique-main}, which lets us restrict our attention to equilibria where each type has threshold equal to their expected utility.)

As a point of comparison when analyzing agents' equilibrium play, we consider the ``idealized thresholds'' $\hat\thresh_{\women_i}$ that arise when type $\women_i$ has absolute market power. These idealized thresholds, which are given by the fixed point equation
\[ \hat\thresh_{\women_i} = \frac{1}{\alpha_{\women_i}} \glow_{\men\women_i}^* \int_{\hat\thresh_{\women_i}}^\infty \utility\,d\CDF_{\men\women_i}, \]
can also be thought of as the threshold that type $\women_i$ agents would set if $\thresh_\men = 0$.
(Note that the left-hand side is monotonically increasing in $\hat\thresh_{\women_i}$ while the right-hand side is monotonically decreasing, so the fixed point exists and is well-defined.) We care about these idealized thresholds because, in equilibria where $\glow_{\men\women_i} = \glow_{\men\women_i}^*$, the threshold chosen by type $\men$ will be a best response to these idealized thresholds. Namely, type $\men$ will play the threshold $\hat\thresh_\men$ which satisfies
\begin{equation}\label{eq:hat-thresh-m}
\hat\thresh_\men
= \frac{1}{\alpha_\men}\sum_{i=1}^\nfill\,\paren*{\glow_{\men\women_i}^*\int_{\max(\hat\thresh_\men, \hat\thresh_{\women_i})}^\infty u \,d\CDF_{\men\women_i}}.
\end{equation}
To prove this claim, we will proceed through the following series of lemmas. These lemmas and the definition of $\hat\thresh_\men$ will also be useful for our later analysis.

For the idealized thresholds $\hat\thresh_{\women_i}$, 
it is not difficult to establish using the definition of $\rho_{\men\women_i}$ and the monotonicity property in \Cref{lemma:rho} that: 
 

\begin{lemma}[Comparison to First-best Thresholds]\label{lemma:thresh-star}
If $i < \nfill$, then $\hat\thresh_{\women_i} = \rho_{\men\women_i}$ (where $\rho_{\men\women_i}$ is defined in \Cref{proposition:lp}). And at $i = \nfill$, it holds that $\hat\thresh_{\women_\nfill}\le\rho_{\men\women_\nfill}$.
\end{lemma}

A similar argument shows that these $\hat\thresh_{\women_i}$ upper bound the equilibrium thresholds $\thresh_{\women_i}$ for all $i$:

\begin{lemma}
[Comparison to Equilibrium Thresholds]
\label{lemma:thresh-ub}
In any equilibrium where $\glow_{\men\women_i} = \glow_{\men\women_i}^*$, $\thresh_{\women_i}\le\hat\thresh_{\women_i}$ for all $i$. 
\end{lemma}

While the thresholds of type $\women_i$ are always upper bounded by $\hat\thresh_{\women_i}$, the quantity $\max(\thresh_\men, \thresh_{\women_i})$ is always at least $\hat\thresh_{\women_i}$: Intuitively, if $\thresh_\men < \hat\thresh_{\women_i}$, then type $\women_i$ would set their threshold to be $\hat\thresh_{\women_i}$, since that would be optimal for them. This property, together with \Cref{lemma:thresh-star}, guarantees that no type $\women_i$ for $i < \nfill$ will match more than they would have under the platform's first-best solution.

\begin{lemma}[Lower Bound on Equilibrium Matching Thresholds]\label{lemma:thresh-max}
In any equilibrium where $\glow_{\men\women_i} = \glow_{\men\women_i}^*$, $\max(\thresh_\men, \thresh_{\women_i})\ge\hat\thresh_{\women_i}$ for all $i$.
\end{lemma}

From \Cref{lemma:thresh-ub,lemma:thresh-max}, we deduce the threshold $\thresh_\men$ played in equilibria where $\glow_{\men\women_i} = \glow_{\men\women_i}^*$ for all $i$: These lemmas together imply that $\max(\thresh_\men, \thresh_{\women_i}) = \max(\thresh_\men, \hat\thresh_{\women_i})$ in such equilibria. Recall that (by \Cref{proposition:unique-main}) $\thresh_\men$ must also satisfy
\begin{equation}
\thresh_\men
= \frac{1}{\alpha_\men}\sum_{i=1}^\nfill\,\paren*{\glow_{\men\women_i}^*\int_{\max(\thresh_\men, \thresh_{\women_i})}^\infty u \,d\CDF_{\men\women_i}} = \frac{1}{\alpha_\men}\sum_{i=1}^\nfill\,\paren*{\glow_{\men\women_i}^*\int_{\max(\thresh_\men, \hat\thresh_{\women_i})}^\infty u \,d\CDF_{\men\women_i}}.
\end{equation}
This is exactly the definition of $\hat\thresh_\men$ as given in \eqref{eq:hat-thresh-m}, so we may conclude $\thresh_\men = \hat\thresh_\men$.




\paragraph{Proving the Proposition.}

With our structural lemmas, we are now ready to prove \Cref{proposition:starshaped}. We split our analysis into two cases based on the value of $\hat\thresh_\men$. We show that if $\hat\thresh_\men > \rho_{\men\women_\nfill}$, then there exists a stationary equilibrium (with incentives) corresponding to the first-best $\glow_{\men\women_i}^*$ values. Moreover, this equilibrium obtains at least half the welfare of the first-best optimal. Otherwise, if $\hat\thresh_\men\le\rho_{\men\women_\nfill}$, then we consider two further possibilities: If $\sum_{i=1}^{\nfill - 1}\rho_{\men\women_i}\cdot\alpha_{\men\women_i}$ is at least half the first-best optimal, then we only let type $\men$ agents meet type $\women_i$ agents for $i < \nfill$. Otherwise, we only let type $\men$ agents meet type $\women_\nfill$ agents. In each of these two remaining cases, we will show that we still capture at least half of the first-best optimal welfare.

\begin{proof}[Proof of \Cref{proposition:starshaped}]
    As mentioned above, we split our analysis into two cases based on the value of $\hat\thresh_\men$. For both, we use $\OPT$ to refer to the welfare value attained by the platform's first-best optimal solution, i.e.,
    $\OPT = 2\sum_{i=1}^\nfill \rho_{\men\women_i}\cdot\beta_{\men\women_i}^*$. We now discuss the two cases:
    
  \begin{description}
    \item[Case 1: $\boldsymbol{\hat\thresh_\men > \rho_{\men\women_\nfill}}$.] If $\hat\thresh_\men > \rho_{\men\women_\nfill}$, we show that the first-best choice of $\glow_{\men\women_i}^*$ induces a stationary equilibrium obtaining at least half the first-best optimal welfare.
    
    To check feasibility, we need to check that the $\{\xi_\type\}_{\type\in\Types}$ and $\{\mass_\type\}_{\type\in\Types}$ given by the thresholds in the preceding lemmas satisfy the constraints of \eqref{eq:opt}. By our analysis in \Cref{proposition:lp}, it is equivalent to check that the combined constraint \eqref{eq:opt-combined-main} holds. By \Cref{lemma:thresh-star,lemma:thresh-max}, in equilibrium, we have $\max(\thresh_\men, \thresh_{\women_i})\ge\hat\thresh_{\women_i} = \rho_{\men\women_i}$ for all $i < \nfill$ and $\thresh_\men = \hat\thresh_\men$, so $\max(\thresh_\men, \thresh_{\women_\nfill})\ge\hat\thresh_{\men} > \rho_{\men\women_\nfill}$. This implies all the thresholds in equilibrium are higher than they were in the platform's first-best solution. In other words, the corresponding matching rates $\xi_\type$ are no larger than they were in the first-best solution. It follows that the combined constraint \eqref{eq:opt-combined-main} can only have more slack, meaning that a stationary  equilibrium is indeed induced by the choice of $\glow_{\men\women_i}^*$ and the thresholds discussed above.
    
    Our remaining task is to check that the realized welfare is sufficient by lower bounding $\hat\thresh_\men$. The intuition for the following calculation is that if type $\men$ is rejecting potential matches of value $\utility$, then the expected payoff of type $\men$ must be at least $\utility$. Indeed, we have:
    \ifx\acmConference\undefined
      \begin{align*}
        \frac 12\OPT
        &= \sum_{\women\in\Women} \rho_{\men\women}\cdot\beta_{\men\women}^* \\
        &= \sum_{\women\in\Women} \paren*{\glow_{\men\women}^* \int_{\rho_{\men\women}}^\infty u \,d\CDF_{\men\women}} \\
        &= \sum_{i=1}^\nfill\, \paren*{\glow_{\men\women_i}^* \int_{\max(\hat\thresh_\men, \rho_{\men\women_i})}^\infty u \,d\CDF_{\men\women_i}} + \sum_{i=1}^\nfill\, \paren*{\glow_{\men\women_i}^* \int_{\rho_{\men\women_i}}^{\max(\hat\thresh_\men, \rho_{\men\women_i})} u \,d\CDF_{\men\women_i}} \\
        &\le \sum_{i=1}^\nfill\, \paren*{\glow_{\men\women_i}^* \int_{\max(\hat\thresh_\men, \hat\thresh_{\women_i})}^\infty u \,d\CDF_{\men\women_i}} + \sum_{i=1}^\nfill\, \paren*{\glow_{\men\women_i}^* \int_{\rho_{\men\women_i}}^{\infty} \hat\thresh_\men \,d\CDF_{\men\women_i}} \\
        &\le \alpha_\men\cdot\frac{1}{\alpha_\men}\sum_{i=1}^\nfill\, \paren*{\glow_{\men\women_i}^* \int_{\max(\hat\thresh_\men, \hat\thresh_{\women_i})}^\infty u \,d\CDF_{\men\women_i}} + \hat\thresh_\men\sum_{i=1}^\nfill \beta^*_{\men\women_i} \\
        &\le 2\alpha_\men\hat\thresh_\men.
      \end{align*}
    \else
      \begin{align*}
        \frac 12\OPT
        &= \sum_{\women\in\Women} \rho_{\men\women}\cdot\beta_{\men\women}^*
        = \sum_{\women\in\Women} \paren*{\glow_{\men\women}^* \int_{\rho_{\men\women}}^\infty u \,d\CDF_{\men\women}} \\
        &= \sum_{i=1}^\nfill\, \paren*{\glow_{\men\women_i}^* \int_{\max(\hat\thresh_\men, \rho_{\men\women_i})}^\infty u \,d\CDF_{\men\women_i}} + \sum_{i=1}^\nfill\, \paren*{\glow_{\men\women_i}^* \int_{\rho_{\men\women_i}}^{\max(\hat\thresh_\men, \rho_{\men\women_i})} u \,d\CDF_{\men\women_i}} \\
        &\le \sum_{i=1}^\nfill\, \paren*{\glow_{\men\women_i}^* \int_{\max(\hat\thresh_\men, \hat\thresh_{\women_i})}^\infty u \,d\CDF_{\men\women_i}} + \sum_{i=1}^\nfill\, \paren*{\glow_{\men\women_i}^* \int_{\rho_{\men\women_i}}^{\infty} \hat\thresh_\men \,d\CDF_{\men\women_i}} \\
        &\le \alpha_\men\cdot\frac{1}{\alpha_\men}\sum_{i=1}^\nfill\, \paren*{\glow_{\men\women_i}^* \int_{\max(\hat\thresh_\men, \hat\thresh_{\women_i})}^\infty u \,d\CDF_{\men\women_i}} + \hat\thresh_\men\sum_{i=1}^\nfill \beta^*_{\men\women_i}
        \le 2\alpha_\men\hat\thresh_\men.
      \end{align*}
    \fi
      The last inequality follows from \eqref{eq:hat-thresh-m} and the constraint on the $\beta_{\men\women}^*$ from \eqref{eq:lp}. Since $2\alpha_\men\hat\thresh_\men$ is the social welfare under the current allocation (where the factor of $2$ comes from accounting for the utilities of both sides), we have the desired lower bound on welfare relative to $\OPT$.
      
    \item[Case 2: $\boldsymbol{\hat\thresh_\men\le \rho_{\men\women_\nfill}}$.]
    When $\hat\thresh_\men\le\rho_{\men\women_\nfill}$, it is possible that the first-best choice of $\glow_{\men\women_i}^*$ may not induce a feasible equilibrium outcome, since type $\women_\nfill$ might accept more than the platform intended in equilibrium, resulting in the combined constraint \eqref{eq:opt-combined-main} being violated for type $\men$. To resolve this, we modify the first-best optimal solution in one of two ways by considering two subcases: In the first subcase, we simply ignore type $\women_\nfill$ and match type $\men$ to types $\women_1,\women_2,\ldots,\women_{\nfill-1}$. In the second subcase, we only match type $\men$ to type $\women_\nfill$. At least one of these two subcases will obtain half the first-best optimal welfare.
    
      \begin{description}
        \item[Subcase 2(a): {$\boldsymbol{\sum_{i=1}^{\nfill-1} \rho_{\men\women_i}\cdot\beta_{\men\women_i}^* > \frac 12\OPT}$}.]
        In this subcase, we consider a modified version of the first-best optimal solution from before, where $\beta_{\men\women_\nfill}$ is set to $0$. We will show that the resulting equilibrium coincides with the first-best equilibrium for the submarket involving only type $\men$ and types $\women_{1},\ldots,\women_{\nfill-1}$.
        
        First, note that such a modification incurs at most a $\frac 12\OPT$ loss in first-best welfare by assumption. Next, recall that type $\men$'s threshold in equilibrium must satisfy the fixed point equation \eqref{eq:hat-thresh-m} for the $\glow_{\men\women_i}^*$ values for $i < \nfill$. It follows that the resulting $\thresh_\men$ is at most $\hat\thresh_\men$ (e.g., via the argument for \Cref{lemma:thresh-ub}).\footnote{Here, we clarify that $\hat\thresh_\men$ is as defined in \eqref{eq:hat-thresh-m} for the $\glow_{\men\women_i}^*$ values for the \emph{unmodified} first-best market.} Since $\hat\thresh_\men\le\rho_{\men\women_\nfill}\le\rho_{\men\women_i} = \hat\thresh_{\women_i}$ for all $i < \nfill$ by \Cref{lemma:first-best,lemma:thresh-star}, we must have $\max(\thresh_\men, \thresh_{\women_i}) = \hat\thresh_{\women_i}$ by \Cref{lemma:thresh-ub,,lemma:thresh-max}. 
          This shows that the resulting equilibrium is identical to the corresponding first-best solution. (In particular, it is feasible.) Hence the social welfare obtained is $\sum_{i=1}^{\nfill-1}\rho_{\men\women_i}\cdot\beta_{\men\women_i}^* > \frac 12\OPT$.
        \item[Subcase 2(b): $\boldsymbol{\rho_{\men\women_\nfill}\cdot\beta_{\men\women_\nfill}^*\ge \frac 12\OPT}$.]
        In this subcase, we consider the first-best solution with $\beta_{\men\women} = 0$ for all $\women\neq\women_\nfill$, and set $\beta_{\men\women} = \min(\alpha_\men, \alpha_{\women_\nfill})$. In this new ``single-edged'' market, where we only match types $\men$ and $\women_{\nfill}$, the first-best social welfare will be at least $\frac 12\OPT$ by assumption.
        In equilibrium, the thresholds must satisfy $\max(\thresh_\men, \thresh_{\women_\nfill}) = \rho_{\men\women_\nfill}$ by appropriate analogs of \Cref{lemma:thresh-star,lemma:thresh-ub,lemma:thresh-max}. Thus, like in Subcase 2(a), the resulting equilibrium is both feasible and identical to that of the corresponding first-best solution and its social welfare is at least $\frac 12\OPT$.\qedhere
      \end{description}
  \end{description}
\end{proof}

 \paragraph{Wrapping Up.}
 
 \Cref{theorem:approximation} now follows straightforwardly from our work up to this point. Since the total of the first-best welfares for the each of the ``star-shaped'' submarkets is at least half the platform's first-best welfare for the overall market and since \Cref{proposition:starshaped} shows that at least half of the first-best welfare for each submarket can be realized in stationary equilibrium, we obtain the desired $4$-approximation for \Cref{theorem:approximation}.

\section{Further Examples via Simulations}
\label{sec:simulation}

To provide further intuition and to give some concrete examples of the model, we present in this section the results of numerical simulations for a set of simple markets. We compare the outcome achieved by our approach to search design against the first-best outcome (defined in \Cref{sec:approximation}) as well as against the baseline outcome of random \vmedit{meeting}.\footnote{While we do not give an algorithm with formal guarantees to compute the outcome of the market under random  \vmedit{meeting}, the outcome can nonetheless be approximated via fixed point iteration. Specifically, we iteratively compute agents' thresholds and equilibrium populations, under the assumption that agents meet uniformly randomly, until a fixed point is reached. This iteration-based approximation converges for all markets considered in this section.} We present two sets of simulations: In the first set of simulations, we interpolate  between the two settings---horizontal and vertical---featured in \Cref{sec:examples}. In the second set of simulations, we highlight the role of search friction (parameterized by $\delta$) on the relative performances of random  \vmedit{meeting}  and designed search.

\subsection{Setup}
\label{subsec:setup}

For both sets of simulations, the market has four types on each side, i.e., $\Men = \{\men_1,\men_2,\men_3,\men_4\}$ and $\Women = \{\women_1, \women_2, \women_3, \women_4\}$. We assume that $(\alpha_{\men_1}, \alpha_{\men_2}, \alpha_{\men_3}, \alpha_{\men_4}) = (1, 4, 4, 1)$ and $(\alpha_{\women_1}, \alpha_{\women_2}, \alpha_{\women_3}, \alpha_{\women_4}) = (2, 2, 3, 2)$. Note that these arrival rates are such that there is slight asymmetry between the types on the two sides.\footnote{This asymmetry leads to more interesting outcomes than if the markets were, say, symmetric. If the markets were symmetric, for our choices of utility values, it would be optimal to fully match type $\men_i$ to type $\women_i$ for all $i$.}

For the first set of experiments, we consider interpolating linearly between a purely horizontal market and a purely vertical market. We assume that each distribution $\Dist_{\men_i\women_j}$ is normal $\mathcal{N}(\mu_{ij}, \sigma^2)$, with mean $\mu_{ij}$ determined linearly interpolating between two extreme horizontal and vertical markets. In the horizontal market, we assume that $\mu_{ij}^h = 0$ if $i\neq j$ and $\mu_{ii}^h = 8$ for all $i$. That is, agents only have positive expected utility from a match if they are of ``corresonding'' types. In the vertical market, we assume that $\mu_{ij}^v = (5 - i)(5 - j)$ for all $i, j$. That is, each type prefers others on the opposite side with smaller indices, with expected utilities being supermodular with respect to this ordering. To interpolate with weight $\weighting\in [0, 1]$, we set $\mu_{ij} = (1-\weighting)\mu_{ij}^h + \weighting\mu_{ij}^v$. Finally, we set $\sigma = 0.1$ and $\delta = 1$. The results of these simulations are plotted in \Cref{fig:experiment1}.

For the second set of experiments, we vary the search friction parameter $\delta$. We run this experiment using the market described above, with weighting $q = 0.5$. To illustrate the role that $\delta$ plays in the relative performance of the approaches, we vary it between $0.01$ and \vmedit{$10$}, evaluating the approaches for $\delta$ belonging to the set $\{0.01, 0.02, 0.05, 0.1, 0.2, 0.5, 1, 2, 5, 10\}$. The results of these simulations are plotted in \Cref{fig:experiment2}.


\begin{figure}
        \centering
        \sffamily
        \FIGURE{{\includegraphics[width = 3in]{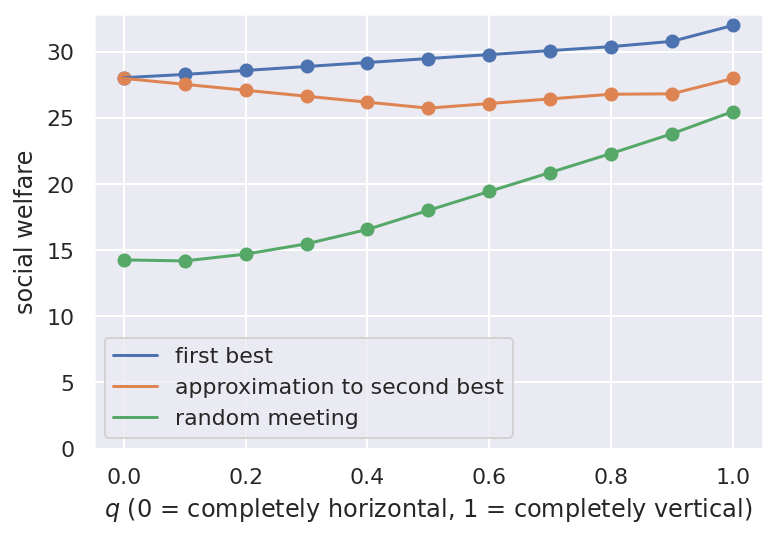}}{\includegraphics[width = 3in]{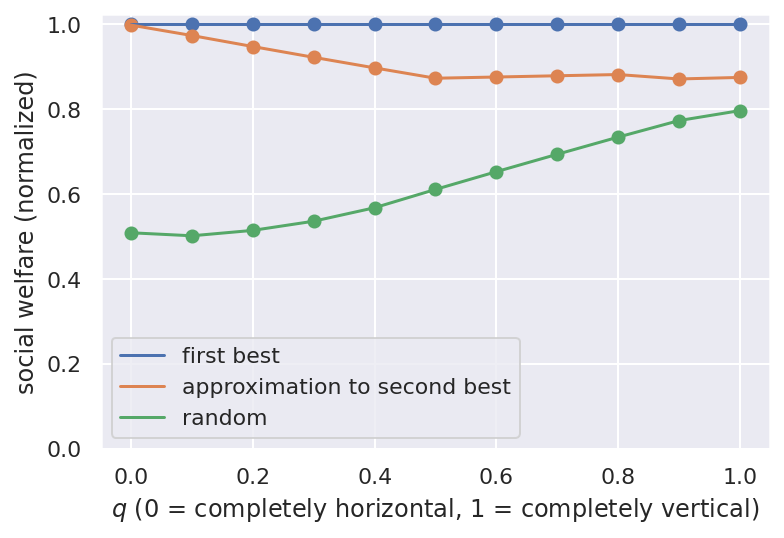}}}{On the left-hand side, we plot social welfare for our approximation algorithm and the two benchmarks of the first best and random meeting as $\weighting$ ranges from $0$ to $1$. On the right-hand side, we plot the same data, but normalize the welfare values so that the first-best social welfare is $1$ for each market.\label{fig:experiment1}}{}
\end{figure}

\subsection{Discussion}

For the first set of simulations (see \Cref{fig:experiment1}), we find that our approximation algorithm has quite good performance relative to the first best and is better than random  \vmedit{meeting}---sometimes significantly so---throughout. The performance gap between random  \vmedit{meeting} and our approximation to the second best is most apparent for the purely horizontal market, corroborating the intuition from \Cref{example:horizontal} that directed search is particularly important when agents' preferences have a strong horizontal component. As the market becomes more vertical, the gap between the performances of our approximation algorithm and random  \vmedit{meeting} narrows. One intuition for why this occurs is that in more vertical markets, while agents match more ``assortatively,'' they also tend to find more types acceptable and thus random  \vmedit{meeting}  can be more successful relative to the stark profile of the horizontal market. Nonetheless, even in such a setting, there is a clear value to directed search (e.g., via our algorithm).

For the second set of simulations (see \Cref{fig:experiment2}), we find that as search friction increases---and thus agents have fewer opportunities to search and match---directed search becomes increasingly necessary to achieve high welfare. Intuitively, such a trend holds because when agents are allowed more opportunities, agents can afford to perform search on their own. On the other hand, when agents only have few opportunities to match, directed search has a larger advantage because the platform can arrange meetings so that agents are able to quickly find desirable matches.

Finally, we note that for the examples considered, our approximation algorithm performs much better than the worst-case $4$-factor approximation proved in \Cref{theorem:approximation}. This is likely because the markets considered for the simulations are relatively simple, with only four types on each side.

\begin{figure}
        \centering
        \sffamily
        \FIGURE{{\includegraphics[width = 3in]{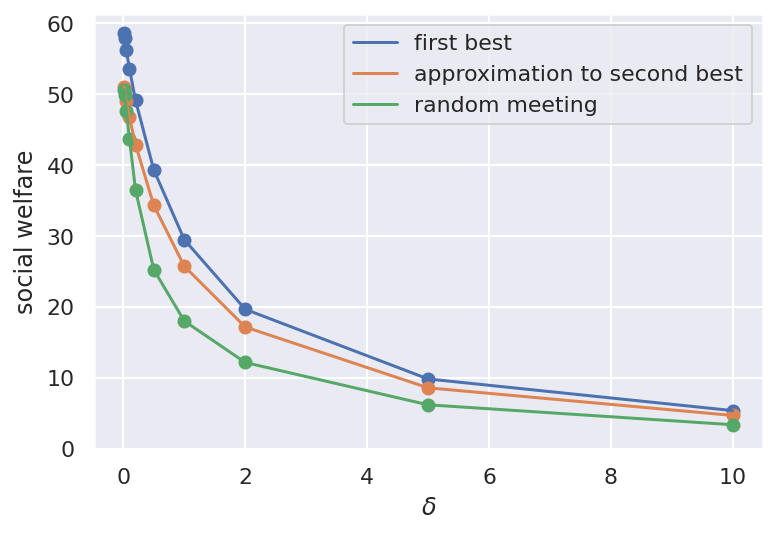}}{\includegraphics[width = 3in]{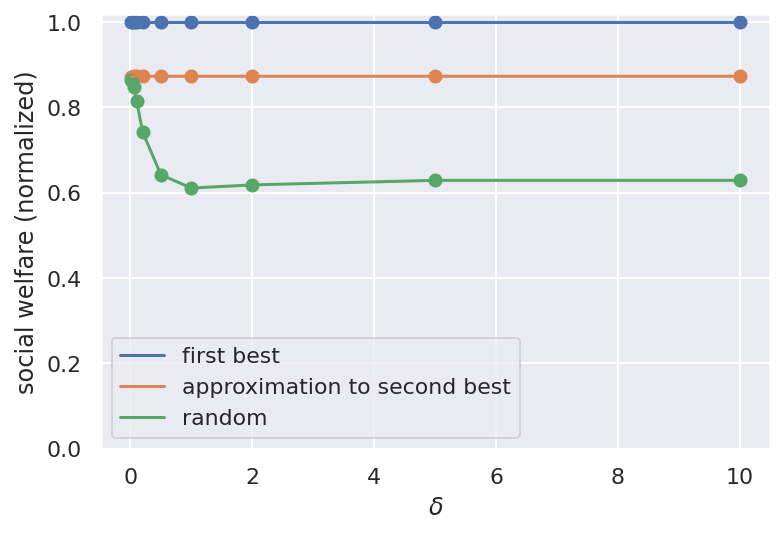}}}{On the left-hand side, we plot social welfare for our approximation algorithm and the two benchmarks of the first best and random meeting as $\delta$ ranges from $0.01$ to $10$. On the right-hand side, we plot the same data, but normalize the welfare values so that the first-best social welfare is $1$ for each market.\label{fig:experiment2}}{}
\end{figure}

\section{Conclusion}\label{sec:conclusion}

Similar to shopping platforms, matching platforms also rely on recommendation systems to facilitate search by offering personalized assortments. However, the two-sided and decentralized nature of these markets makes the design of their recommendation systems fundamentally different from those used for product recommendation. 
Congestion and misaligned incentives often necessitate making recommendations that are sub-optimal for certain agents but improve the overall social welfare. 
In this work, we take a first step toward understanding the intricacies of designing recommendation systems based on imperfect knowledge about preferences while taking agents' strategic behavior into consideration.
Somewhat surprisingly, we show that for general symmetric preferences, carefully designed assortments with very limited choice can achieve approximately optimal welfare.  

\vmedit{While our framework is general in terms of agent types and the structure of symmetric preferences, it abstracts away from some  considerations, which we discuss next. In our model, we assume that once the platform suggests a meeting between two agents, those agents meet, observe the utility of matching with each other, and make a decision to accept or reject the match. However, since meeting---which corresponds to observing the realized utility via a screening action---can be costly, one might consider a richer model (similar to that of \citet{KanoriaSaban21}) in which agents decide whether to (i) accept/reject the match without meeting, (ii) meet at a screening cost and then make an accept/reject decision. 
Studying search design under such a model for general preference structures is a valuable direction for future work. 
Incorporating screening cost would involve enriching the set of actions for each agent. This, in turn, would make encoding agents' strategic behavior in the design problem more complex.} 
\vmedit{Furthermore, the tractability of our framework as well as our design problem  crucially relies on having symmetric preferences and a stationary environment. As such, incorporating asymmetric match utilities or non-stationary behavior---such as history-dependent departure rates or strategies---would require new modeling/technical developments and is an interesting direction for future research. 
It is well-known that even for simple asymmetric preferences and under random meeting, multiple stationary equilibria may arise \citep{BurdettColes97}.  Thus, the platform design problem may go beyond designing meeting distributions and involve mechanisms for equilibrium selection.}

\vmedit{We conclude by noting that in this work, we focus on the design of matching markets with nontransferable utilities. Such models, also known as marriage models, are motivated by online dating platforms or labor markets with fixed wages.
They
can also be viewed as a special case of markets with transferable utilities in which the match utility is evenly split between the two parties \citep{smith2006marriage}. Studying a similar setting with transfers---in which the two agents' payoffs constitute the Nash bargaining solution---is also well-motivated in the context of online labor markets where wages are negotiable. In such a setting, the platform's first-best solution remains unchanged. As a future direction, one could explore the possibility of building on our hardness result and approximation scheme to establish analogous results for designing search in matching markets with transferable utilities.}



\setlength{\bibsep}{0.0pt}
\def\newblock{}
\OneAndAHalfSpacedXI
{
\bibliographystyle{plainnat}
\bibliography{bibliography.bib}}

\begin{thebibliography}{42}
\providecommand{\natexlab}[1]{#1}
\providecommand{\url}[1]{\texttt{#1}}
\expandafter\ifx\csname urlstyle\endcsname\relax
  \providecommand{\doi}[1]{doi: #1}\else
  \providecommand{\doi}{doi: \begingroup \urlstyle{rm}\Url}\fi

\bibitem[Abraham et~al.(2008)Abraham, Levavi, Manlove, and
  O'Malley]{abraham2008stable}
David~J Abraham, Ariel Levavi, David~F Manlove, and Gregg O'Malley.
\newblock The stable roommates problem with globally ranked pairs.
\newblock \emph{Internet Mathematics}, 5\penalty0 (4):\penalty0 493--515, 2008.

\bibitem[Ackermann et~al.(2011)Ackermann, Goldberg, Mirrokni, R{\"o}glin, and
  V{\"o}cking]{ackermann2011uncoordinated}
Heiner Ackermann, Paul~W Goldberg, Vahab~S Mirrokni, Heiko R{\"o}glin, and
  Berthold V{\"o}cking.
\newblock Uncoordinated two-sided matching markets.
\newblock \emph{SIAM Journal on Computing}, 40\penalty0 (1):\penalty0 92--106,
  2011.

\bibitem[Adachi(2003)]{adachi2003search}
Hiroyuki Adachi.
\newblock A search model of two-sided matching under nontransferable utility.
\newblock \emph{Journal of Economic Theory}, 113\penalty0 (2):\penalty0
  182--198, 2003.

\bibitem[Afeche et~al.(2021)Afeche, Caldentey, and Gupta]{afeche2021optimal}
Philipp Afeche, Rene Caldentey, and Varun Gupta.
\newblock On the optimal design of a bipartite matching queueing system.
\newblock \emph{Operations Research}, 2021.

\bibitem[Akbarpour et~al.(2020)Akbarpour, Li, and
  Gharan]{akbarpour2020thickness}
Mohammad Akbarpour, Shengwu Li, and Shayan~Oveis Gharan.
\newblock Thickness and information in dynamic matching markets.
\newblock \emph{Journal of Political Economy}, 128\penalty0 (3):\penalty0
  783--815, 2020.

\bibitem[Anunrojwong et~al.(2021)Anunrojwong, Iyer, and
  Manshadi]{anunrojwong2021information}
Jerry Anunrojwong, Krishnamurthy Iyer, and Vahideh Manshadi.
\newblock Information design for congested social services: Optimal need-based
  persuasion.
\newblock \emph{Available at SSRN 3849746}, 2021.

\bibitem[Aouad and Saban(2020)]{aouad2020online}
Ali Aouad and Daniela Saban.
\newblock Online assortment optimization for two-sided matching platforms.
\newblock \emph{Available at SSRN}, 2020.

\bibitem[Arnosti and Shi(2020)]{arnosti2020design}
Nick Arnosti and Peng Shi.
\newblock Design of lotteries and wait-lists for affordable housing allocation.
\newblock \emph{Management Science}, 66\penalty0 (6):\penalty0 2291--2307,
  2020.

\bibitem[Ashlagi et~al.(2013)Ashlagi, Jaillet, and Manshadi]{AshlagiKE}
Itai Ashlagi, Patrick Jaillet, and Vahideh~H. Manshadi.
\newblock Kidney exchange in dynamic sparse heterogenous pools.
\newblock In \emph{Proceedings of the Fourteenth ACM Conference on Electronic
  Commerce}, EC '13, page 25–26, New York, NY, USA, 2013. Association for
  Computing Machinery.
\newblock ISBN 9781450319621.
\newblock \doi{10.1145/2482540.2482565}.
\newblock URL \url{https://doi.org/10.1145/2482540.2482565}.

\bibitem[Ashlagi et~al.(2019{\natexlab{a}})Ashlagi, Burq, Jaillet, and
  Manshadi]{ashlagi2019matching}
Itai Ashlagi, Maximilien Burq, Patrick Jaillet, and Vahideh Manshadi.
\newblock On matching and thickness in heterogeneous dynamic markets.
\newblock \emph{Operations Research}, 67\penalty0 (4):\penalty0 927--949,
  2019{\natexlab{a}}.

\bibitem[Ashlagi et~al.(2019{\natexlab{b}})Ashlagi, Krishnaswamy, Makhijani,
  Saban, and Shiragur]{ashlagi2019assortment}
Itai Ashlagi, Anilesh~K Krishnaswamy, Rahul Makhijani, Daniela Saban, and
  Kirankumar Shiragur.
\newblock Assortment planning for two-sided sequential matching markets.
\newblock \emph{arXiv preprint arXiv:1907.04485}, 2019{\natexlab{b}}.

\bibitem[Ashlagi et~al.(2020)Ashlagi, Braverman, Kanoria, and
  Shi]{ashlagi2020clearing}
Itai Ashlagi, Mark Braverman, Yash Kanoria, and Peng Shi.
\newblock Clearing matching markets efficiently: informative signals and match
  recommendations.
\newblock \emph{Management Science}, 66\penalty0 (5):\penalty0 2163--2193,
  2020.

\bibitem[Ashlagi et~al.(2021)Ashlagi, Monachou, and Nikzad]{ashlagi2021optimal}
Itai Ashlagi, Faidra Monachou, and Afshin Nikzad.
\newblock Optimal dynamic allocation: Simplicity through information design.
\newblock In \emph{Proceedings of the 22nd ACM Conference on Economics and
  Computation}, pages 101--102, 2021.

\bibitem[Banerjee et~al.(2017)Banerjee, Gollapudi, Kollias, and
  Munagala]{banerjee2017segmenting}
Siddhartha Banerjee, Sreenivas Gollapudi, Kostas Kollias, and Kamesh Munagala.
\newblock Segmenting two-sided markets.
\newblock In \emph{Proceedings of the 26th International Conference on World
  Wide Web, {WWW} 2017, Perth, Australia, April 3-7, 2017}, pages 63--72.
  {ACM}, 2017.

\bibitem[Burdett and Coles(1997)]{BurdettColes97}
Ken Burdett and Melvyn~G. Coles.
\newblock Marriage and class.
\newblock \emph{The Quarterly Journal of Economics}, 112\penalty0 (1):\penalty0
  141--168, 1997.

\bibitem[Chade(2006)]{chade2006matching}
Hector Chade.
\newblock Matching with noise and the acceptance curse.
\newblock \emph{Journal of Economic Theory}, 129\penalty0 (1):\penalty0
  81--113, 2006.

\bibitem[Chade et~al.(2017)Chade, Eeckhout, and Smith]{chade2017sorting}
Hector Chade, Jan Eeckhout, and Lones Smith.
\newblock Sorting through search and matching models in economics.
\newblock \emph{Journal of Economic Literature}, 55\penalty0 (2):\penalty0
  493--544, 2017.

\bibitem[Chakrabarty and Goel(2010)]{ChakrabartyGoel10}
Deeparnab Chakrabarty and Gagan Goel.
\newblock On the approximability of budgeted allocations and improved lower
  bounds for submodular welfare maximization and {GAP}.
\newblock \emph{{SIAM} J. Comput.}, 39\penalty0 (6):\penalty0 2189--2211, 2010.

\bibitem[Che and Tercieux(2021)]{che2021optimal}
Yeon-Koo Che and Olivier Tercieux.
\newblock Optimal queue design.
\newblock In \emph{Proceedings of the 22nd ACM Conference on Economics and
  Computation}, pages 312--313, 2021.

\bibitem[Coles and Shorrer(2014)]{coles2014optimal}
Peter Coles and Ran Shorrer.
\newblock Optimal truncation in matching markets.
\newblock \emph{Games and Economic Behavior}, 87:\penalty0 591--615, 2014.

\bibitem[Doval and Szentes(2019)]{doval2019efficiency}
Laura Doval and Bal{\'a}sz Szentes.
\newblock On the efficiency of queueing in dynamic matching markets.
\newblock Technical report, Working paper, Division of the Humanities and
  Social Sciences, California Institute of Technology, 2019.

\bibitem[Duffie et~al.(2018)Duffie, Qiao, and Sun]{duffie2018dynamic}
Darrell Duffie, Lei Qiao, and Yeneng Sun.
\newblock Dynamic directed random matching.
\newblock \emph{Journal of Economic Theory}, 174:\penalty0 124--183, 2018.

\bibitem[Emamjomeh-Zadeh et~al.(2020)Emamjomeh-Zadeh, Gonczarowski, and
  Kempe]{emamjomeh2020complexity}
Ehsan Emamjomeh-Zadeh, Yannai~A Gonczarowski, and David Kempe.
\newblock The complexity of interactively learning a stable matching by trial
  and error.
\newblock \emph{arXiv preprint arXiv:2002.07363}, 2020.

\bibitem[Feigenbaum et~al.(2020)Feigenbaum, Kanoria, Lo, and
  Sethuraman]{feigenbaum2020dynamic}
Itai Feigenbaum, Yash Kanoria, Irene Lo, and Jay Sethuraman.
\newblock Dynamic matching in school choice: Efficient seat reassignment after
  late cancellations.
\newblock \emph{Management Science}, 66\penalty0 (11):\penalty0 5341--5361,
  2020.

\bibitem[Gonczarowski et~al.(2019)Gonczarowski, Nisan, Ostrovsky, and
  Rosenbaum]{gonczarowski2019stable}
Yannai~A Gonczarowski, Noam Nisan, Rafail Ostrovsky, and Will Rosenbaum.
\newblock A stable marriage requires communication.
\newblock \emph{Games and Economic Behavior}, 118:\penalty0 626--647, 2019.

\bibitem[Halaburda et~al.(2018)Halaburda, Piskorski, and
  Yildirim]{halaburda2016competing}
Hanna Halaburda, Mikolaj~Jan Piskorski, and Pinar Yildirim.
\newblock Competing by restricting choice: The case of matching platforms.
\newblock \emph{Management Science}, 64\penalty0 (8):\penalty0 3574--3594,
  2018.

\bibitem[H{\aa}stad(2001)]{Hastad01}
Johan H{\aa}stad.
\newblock Some optimal inapproximability results.
\newblock \emph{J. {ACM}}, 48\penalty0 (4):\penalty0 798--859, 2001.

\bibitem[Immorlica and Mahdian(2015)]{immorlica2015incentives}
Nicole Immorlica and Mohammad Mahdian.
\newblock Incentives in large random two-sided markets.
\newblock \emph{ACM Transactions on Economics and Computation (TEAC)},
  3\penalty0 (3):\penalty0 1--25, 2015.

\bibitem[Immorlica et~al.(2020)Immorlica, Leshno, Lo, and
  Lucier]{immorlica2020information}
Nicole Immorlica, Jacob Leshno, Irene Lo, and Brendan Lucier.
\newblock Information acquisition in matching markets: The role of price
  discovery.
\newblock \emph{Available at SSRN}, 2020.

\bibitem[Kanoria and Saban(2021)]{KanoriaSaban21}
Yash Kanoria and Daniela Saban.
\newblock Facilitating the search for partners on matching platforms.
\newblock \emph{Management Science}, page to appear, 2021.

\bibitem[Lauermann and Nöldeke(2014)]{LauermannNoldeke14}
Stephan Lauermann and Georg Nöldeke.
\newblock Stable marriages and search frictions.
\newblock \emph{Journal of Economic Theory}, 151:\penalty0 163--195, 2014.

\bibitem[Lenstra et~al.(1990)Lenstra, Shmoys, and
  Tardos]{LenstraShmoysTardos90}
Jan~Karel Lenstra, David~B. Shmoys, and {\'{E}}va Tardos.
\newblock Approximation algorithms for scheduling unrelated parallel machines.
\newblock \emph{Mathematical Programming}, 46:\penalty0 259--271, 1990.

\bibitem[Liu et~al.(2020)Liu, Mania, and Jordan]{liu2020competing}
Lydia~T. Liu, Horia Mania, and Michael~I. Jordan.
\newblock Competing bandits in matching markets.
\newblock In \emph{The 23rd International Conference on Artificial Intelligence
  and Statistics, {AISTATS} 2020, 26-28 August 2020, Online [Palermo, Sicily,
  Italy]}, volume 108 of \emph{Proceedings of Machine Learning Research}, pages
  1618--1628. {PMLR}, 2020.

\bibitem[Nikzad(2017)]{nikzad2017thickness}
Afshin Nikzad.
\newblock Thickness and competition in ride-sharing markets.
\newblock \emph{Available at SSRN 3065672}, 2017.

\bibitem[Ozimek(2019)]{ozimek2019}
Adam Ozimek.
\newblock Report: Freelancing and the economy in 2019.
\newblock 2019.

\bibitem[{Pew Research Center}(2020)]{pew2020}
{Pew Research Center}.
\newblock The virtues and downsides of online dating.
\newblock 2020.

\bibitem[R{\'\i}os et~al.(2020)R{\'\i}os, Saban, and Zheng]{rios2020improving}
Ignacio R{\'\i}os, Daniela Saban, and Fanyin Zheng.
\newblock Improving match rates in dating markets through assortment
  optimization.
\newblock \emph{Available at SSRN}, 2020.

\bibitem[Roth(1982)]{roth1982economics}
Alvin~E Roth.
\newblock The economics of matching: Stability and incentives.
\newblock \emph{Mathematics of operations research}, 7\penalty0 (4):\penalty0
  617--628, 1982.

\bibitem[Shimer and Smith(2000)]{shimer2000assortative}
Robert Shimer and Lones Smith.
\newblock Assortative matching and search.
\newblock \emph{Econometrica}, 68\penalty0 (2):\penalty0 343--369, 2000.

\bibitem[Shimer and Smith(2001)]{shimer2001matching}
Robert Shimer and Lones Smith.
\newblock Matching, search, and heterogeneity.
\newblock \emph{The BE Journal of Macroeconomics}, 1\penalty0 (1), 2001.

\bibitem[Smith(2006)]{smith2006marriage}
Lones Smith.
\newblock The marriage model with search frictions.
\newblock \emph{Journal of political Economy}, 114\penalty0 (6):\penalty0
  1124--1144, 2006.

\bibitem[Sun(2006)]{sun2006elln}
Yeneng Sun.
\newblock The exact law of large numbers via {Fubini} extension and
  characterization of insurable risks.
\newblock \emph{Journal of Economic Theory}, 126\penalty0 (1):\penalty0 31--69,
  2006.

\end{thebibliography}

\newpage
\renewcommand{\theHsection}{A\arabic{section}}
\begin{APPENDIX}{}


\section{Formal Definitions of MDP, Strategies, and Equilibrium}\label{sec:formalities}

This section develops the mathematical formalization for the agents' decision problem introduced in \Cref{sec:strategies}. We also prove \Cref{lemma:payoff} and \Cref{proposition:unique-main} using the formalization developed here.

\subsection{Solving for the Agent's Best Response}\label{sec:mdp}

Recall from \Cref{sec:strategies} that any fixed set of assortments defines a game for the agents, in which agents decide who to accept and reject in their meetings. Recall also that agents play symmetric time- and history-independent strategies. We thus write the strategy of type $\type$ agents as a function $\strategy_\type(\typeii, \utility)$ taking values in $[0, 1]$ which specifies the probability that a type $\type$ agent accepts when meeting a type $\typeii$ agent whom they value at utility $\utility$. Then $\strategy_\type(\typeii, \utility)\cdot\strategy_\typeii(\type, \utility)$ is the probability that two agents of types $\type$ and $\typeii$ who value each other at utility $\utility$ mutually accept. For measure-theoretic purposes, we restrict our attention to strategies $\strategy_\type$ that are measurable as a function in $\utility$ for all $\typeii$.

Rather than modeling a full-fledged dynamic game, we model each type $\type$ agent as facing a continuous-time MDP defined in terms of their (time-invariant) assortment $\rate_\type$ and opposing types' (time-invariant) strategy profiles $\strategy_\typeii$. We may do this because, as noted in \Cref{sec:strategies}, agents are only affected by the actions of other types in aggregate; thus any agent's own strategic decision will not affect the (aggregate) action of others.
In this MDP, the agent starts in a ``waiting'' state. From here, they either transition to an ``exited'' state when they leave unmatched due to a life event or to a ``meeting'' state when they meet another agent. At this ``meeting'' state, the agent makes a decision to either ``accept'' or ``reject''. Then, if they match, they transition to ``exited'' with a payoff; otherwise, they transition back to ``waiting.''

That life events and meetings occur in a memoryless manner actually means we can rid ourselves of the continuous-time aspect of this MDP and get an equivalent discrete-time MDP.
Formally, we set up this discrete-time MDP as follows: The agent's initial state is $\Waiting$. State $\Waiting$ transitions to $\Exited$ with probability $\delta / (\delta + \sum_\typeii \rate_\type(\typeii))$ and to $\PreMeeting_\typeii$ with probability $\rate_\type(\typeii) / (\delta + \sum_\typeii \rate_\type(\typeii))$. State $\PreMeeting_\typeii$ transitions to $\Meeting_{\typeii,\utility}$ with $\utility$ drawn from $\Dist_{\type\typeii}$. State $\Meeting_{\typeii,\utility}$ is a decision point where the agent can either $\Accept$ or $\Reject$. If $\Accept$ is chosen, with probability $\sigma_\typeii(\type, \utility)$ they transition to $\Exited$ and receive payoff $\utility$. In all other cases, they return to state $\Waiting$. Finally, $\Exited$ is a terminal state.

The optimal policy for this MDP admits a simple informal analysis: If $\contutility_\type$ is the expected payoff of an agent in state $\Waiting$, then $\contutility_\type$ is also the expected continuation payoff if an agent chooses $\Reject$ in state $\Meeting_{\typeii,\utility}$. Hence the agent should choose $\Accept$ if $\utility > \contutility_\type$ and $\Reject$ if $\utility < \contutility_\type$. If $\utility = \contutility_\type$, the agent is indifferent between the two options; however, since the distributions $\Dist_{\men\women}$ are continuous, this occurs with probability $0$. We now state this observation more formally:

\begin{lemma}\label{lemma:weakthreshold}
For any best response $\strategy_\type$, let $\contutility_\type$ be its expected payoff. Then
\[ \strategy_\type(\typeii, \utility) = \begin{cases} 1 & \text{if $\utility > \contutility_\type$} \\ 0 & \text{if $\utility < \contutility_\type$}\end{cases} \]
for all $\typeii$ such that $\rate_\type(\typeii) > 0$ and almost all $\utility$ in the support of $\strategy_\typeii(\type, \utility)\,d\CDF_{\type\typeii}$.
\end{lemma}

Proving such a claim is standard in the analysis of MDPs, so we omit a detailed formal proof.

\Cref{lemma:weakthreshold} delineates a collection of dominated strategies, namely those that reject potential matches that provide greater utility than the continuation utility (i.e., strategies $\strategy_\type$ such that $\strategy_\type(\typeii, \utility) < 1$ for some $\typeii$ and $\utility > \contutility_\type$). In the subsequent analysis, we rule out such dominated strategies to eliminate ``bad'' equilibria, e.g., when all agents reject all potential matches (knowing that no one else ever accepts).

Ruling out dominated strategies as above in fact lets us focus on equilibria in \emph{threshold strategies}, where each type has a threshold $\thresh_\type$ and accepts if and only if $\utility\ge\thresh_\type$. Our next lemma shows that any equilibrium in non-dominated strategies is equivalent, in a certain sense, to an equilibrium in threshold strategies, where each type's threshold is also their expected utility:

\begin{lemma}\label{lemma:threshold}
Suppose all agents are playing best responses where they accept all potential matches worth at least their expected utility. If type $\type$ agents all switch to playing the threshold strategy with threshold $\thresh_\type = \contutility_\type$, where $\contutility_\type$ is the expected utility of type $\type$ agents, then each type's strategy is still a best response. Moreover, the distribution of matches and the utilities at which they are realized remains the same (up to measure $0$).
\end{lemma}

\begin{proof}
Since type $\type$ agents already accept all potential matches worth at least $\contutility_\type$, by switching to the threshold strategy with threshold $\contutility_\type$, the only change is that they now reject all potential matches worth less than $\contutility_\type$. By \Cref{lemma:weakthreshold}, accepting such potential matches was already a probability $0$ event to begin with. Hence the distribution of realized matches remains the same.

To see that all strategies are still best responses, notice that the expected utility of type $\type$ agents does not change. Furthermore, type $\type$ agents rejecting more matches only restricts the choice sets of agents on the opposite side. Hence their strategies $\strategy_\typeii$ remain best responses also.
\end{proof}

Applying \Cref{lemma:threshold} in succession to all types $\type\in\Types$ allows us to convert any equilibrium to an equivalent equilibrium in threshold strategies, such that each agent is thresholding at their expected utility. Thus, without loss of generality, we may restrict our attention to such equilibria.

We can also show a converse of sorts to \Cref{lemma:threshold}, in which we characterize the optimal threshold for type $\type$ as a solution to a fixed point equation:

\begin{lemma}\label{lemma:fixedpoint}
There exists a unique fixed point satisfying $\thresh_\type = \contutility_\type$ (viewing $\contutility_\type$ as a function of $\thresh_\type$). In particular, if $\thresh_\type$ satisfies the fixed point equation, then $\thresh_\type$ is the unique threshold best response for type $\type$ agents such that $\thresh_\type = \contutility_\type$.
\end{lemma}

\begin{proof}
Let $\thresh_\type$ be a best response threshold satisfying the fixed point equation. (By \Cref{lemma:threshold}, such a $\thresh_\type$ exists.) Note that no threshold $\threshii_\type > \thresh_\type$ can satisfy the fixed point equation because $\thresh_\type$ is a best response. So suppose $\threshii_\type < \thresh_\type$. Then, an agent thresholding at $\threshii_\type$ would leave the market no later than an agent thresholding at $\thresh_\type$. If they leave strictly earlier, then they must have left with payoff at least $\threshii_\type$. And conditioned on leaving at the same time, their expected payoff is $\thresh_\type$. The latter occurs with positive probability, so their expected payoff when thresholding at $\threshii_\type$ exceeds $\threshii_\type$. It follows that $\threshii_\type$ is not a fixed point, making $\thresh_\type$ the unique fixed point.
\end{proof}




\subsection{Proof of \Cref{lemma:payoff}}

\begin{proof}[Proof of \Cref{lemma:payoff}]
To compute the agent's expected payoff at state $\Waiting$, we can condition on the event that the agent does not return to $\Waiting$. Then, the conditional probability of matching with an agent of type $\typeii$ is 
\[ \frac{\lambda_\type(\typeii)\int_{\max(\thresh_\type, \thresh_\typeii)}^\infty \,d\CDF_{\type\typeii}}{\delta + \sum_\typeii \paren*{\rate_\type(\typeii) \int_{\max(\thresh_\type, \thresh_\typeii)}^\infty \,d\CDF_{\type\typeii}}} \]
and the expected payoff conditioned on matching with an agent of type $\typeii$ is
\[ \frac{\int_{\max(\thresh_\type, \thresh_\typeii)}^\infty \utility\,d\CDF_{\type\typeii}}{\int_{\max(\thresh_\type, \thresh_\typeii)}^\infty \,d\CDF_{\type\typeii}}. \]
Summing over all $\typeii$ the product of the above two expression gives us the first equality. An application of \Cref{lemma:rate} along with the definition $\zeta_\type = \xi_\type + \delta$ gives us the second \vmedit{statement}.
\end{proof}

\subsection{Proof of \Cref{proposition:unique-main}: Equilibrium Play for Fixed Assortments}\label{sec:formalities-unique}

In this section, we complete the proof \Cref{proposition:unique-main}, which states that our model makes a unique prediction of equilibrium play given any fixed set of assortments. We show that there exists a strategy profile where each strategy is a best response and that this strategy profile is unique. Uniqueness implies that stationary equilibria are robustly self-sustaining (e.g., a market in stationary equilibrium cannot be disrupted by agents switching to another strategy profile of best responses) and that the social welfare of an assortment is well-defined (see \Cref{sec:optimaldirectedsearch}). Our proof of the existence of a profile of best responses is constructive and can be implemented as an algorithm to compute equilibrium play. Consequently, our model's unique prediction of equilibrium play can also be efficiently computed by the platform.

To prove that a profile of best responses exists, we show that iterating the best response map converges after finitely many iterations. (In a sense, this iteration behaves like a Gale-Shapley operator.) Our uniqueness result derives from the symmetric valuations assumption that we make (see \Cref{sec:utilities}): Symmetric valuations induce a linear ordering over all possible matches by their valuations. This rules out ``cycles'' in the preferences and thus the possibility of multiple equilibria.

More formally, let the best response map take as input a strategy profile $\{\strategy_\type\}_{\type\in\Types}$ and output a new strategy profile $\{\strategyii_\type\}_{\type\in\Types}$ such that each $\strategyii_\type$ is the threshold strategy that thresholds at the expected utility of the agent under the input strategy profile. (This map is a best response by \Cref{lemma:fixedpoint}.) To show existence of equilibrium, we will show that iterating this map on the strategy profile converges.


\begin{proposition}\label{proposition:unique}
Given assortments $\{\rate_\type\}_{\type\in\Types}$, there is a unique profile $\{\strategy_\type\}_{\type\in\Types}$ in non-dominated strategies such that each $\strategy_\type$ is a best response (up to the equivalence given in \Cref{lemma:threshold}). This strategy profile can be found by iterating the best-response map $O(|\Types|)$ times.
\end{proposition}

Before we proceed with the proof of \Cref{proposition:unique}, we note that the claims of \Cref{proposition:unique-main} follow from combining \Cref{lemma:threshold,lemma:fixedpoint} with \Cref{proposition:unique}.

\begin{proof}
By applying \Cref{lemma:threshold}, we may restrict our attention to agents playing threshold strategies where each type thresholds at their expected utility $\contutility_\type$.

To show existence, we iterate the best response map starting from the strategy profile where each type's threshold is $0$. We show that this iteration converges after $O(\abs{\Types})$ steps. In the sequel, when we refer to the thresholds of the $i$-th iteration, we mean the thresholds after applying the best response map $i$ times.

Our first step is to show that each type's threshold has the following pattern: The sequence of thresholds after an odd number of iterations is (weakly) monotonically decreasing and that the sequence of thresholds after an even number of iterations is (weakly) monotonically increasing. We prove this claim by induction. For the base case, note that the claim is trivially true when comparing the zeroth and second iterations, since all thresholds are initialized to $0$. Next, suppose the claim is true for iterations $i - 3$ and $i - 1$. Then, note that the thresholds after iterations $i - 2$ and $i$ are each given by the expected utilities of agents. Now, if $i$ is odd, then we know that thresholds after iteration $i - 1$ are higher than those after iteration $i - 3$; hence, the expected utilities, and therefore the new thresholds, are uniformly lower. The analogous argument works for even $i$.

A similar argument shows that for each type, the infimum of that type's thresholds after odd-numbered iterations is at least the supremum of that type's thresholds after even-numbered iterations: The claim clearly holds for the initialization of all thresholds at $0$. Next, suppose the claim holds for iteration $i - 1$. If $i$ is odd, then let $j > i$ be any even number. (Note that the claim follows from the monotonicity property above for $j < i$.) The thresholds after the $j$-th iteration are best responses to thresholds after the $(j-1)$-th iteration. By the inductive hypothesis, all the thresholds after the $(j-1)$-th iteration are at least those of the $(i-1)$-th iteration. Hence the expected utility of each type in the $(j-1)$-th iteration must be at most that of each type in the $(i-1)$-th iteration. Hence the thresholds after the $j$-th iteration are uniformly \vmedit{smaller} than the thresholds after the $i$-th iteration. The analogous argument works for even $i$.

The final step for showing the convergence of this best response iteration is to show that after each odd-numbered iteration, at least one additional type will have their threshold ``frozen,'' meaning that it will not change in any future iterations. Indeed, after any odd-numbered iteration, consider the agent type $\type$ with the highest threshold $\thresh_\type$ who has not yet been shown to be frozen. Then, observe that in all future iterations, no unfrozen agent will have a higher threshold by the two claims above---the first claim handles odd-numbered iterations and the second claim handles even-numbered iterations. It follows that type $\type$ would not want to alter its best response in all following iterations: It will not be affected by the strategy of any unfrozen type, because their thresholds will be at most $\thresh_\type$; the strategies of all frozen types will remain the same. Thus, type $\type$ will be frozen after this iteration. Since at least one additional type has their threshold frozen after each odd-numbered iteration, iterating the best response map will converge after $O(\abs{\Types})$ iterations.


Next, we show the uniqueness of this equilibrium. Suppose for the sake of contradiction that there are two distinct strategy profiles $\{\strategy_{\type}\}_{\type\in\Types}$ and $\{\strategyii_{\type}\}_{\type\in\Types}$ such that each strategy is both a best response and a threshold strategy. Let the two sets of thresholds for these two strategy profiles be $\{\thresh_\type\}_{\type\in\Types}$ and $\{\threshii_\type\}_{\type\in\Types}$. Furthermore, suppose each threshold is its corresponding type's expected utility. Since the two strategy profiles are distinct, there exists a type $\type$ such that $\thresh_\type\neq\threshii_\type$ and $\max(\thresh_\type, \threshii_\type)$ is maximal.
Without loss of generality, we may assume that $\thresh_\type > \threshii_\type$. It follows that if $\thresh_{\typeii} > \thresh_\type$, then $\thresh_{\typeii} = \threshii_{\typeii}$, and if $\thresh_\type\ge\thresh_\typeii$, then $\thresh_\type\ge\threshii_\typeii$. Consequently, $\max(\thresh_\type, \thresh_\typeii) = \max(\thresh_\type, \threshii_\typeii)$ for all $\typeii$. But this gives us the following contradiction: Threshold $\thresh_\type$ is type $\type$'s expected utility, so \Cref{lemma:payoff} tells us
\[ \threshii_\type < \thresh_\type = \frac{\sum_{\typeii}\paren[\Big]{\rate_{\type}(\typeii)\int_{\max(\thresh_\type, \thresh_\typeii)}^\infty\utility \,d\CDF_{\type\typeii}}}{\delta + \sum_{\typeii}\paren[\Big]{\rate_{\type}(\typeii)\int_{\max(\thresh_\type, \thresh_\typeii)}^\infty \,d\CDF_{\type\typeii}}} = \frac{\sum_{\typeii}\paren[\Big]{\rate_{\type}(\typeii)\int_{\max(\thresh_\type, \threshii_\typeii)}^\infty\utility \,d\CDF_{\type\typeii}}}{\delta + \sum_{\typeii}\paren[\Big]{\rate_{\type}(\typeii)\int_{\max(\thresh_\type, \threshii_\typeii)}^\infty \,d\CDF_{\type\typeii}}}. \]
That the right-hand side exceeds the left-hand side means thresholding at $\threshii_\type$ is not actually a best response, since we could do better by choosing our threshold to be $\thresh_\type$. Therefore, two such distinct strategy profiles cannot exist, proving the proposition.
\end{proof}

%

\section{Omitted Proofs from Section \ref{sec:approximation}}\label{sec:approximation-proofs}

\subsection{Proof of \Cref{proposition:first-second-gap}}
\label{app:proof:first-second-gap}

\begin{proof}[Proof of \Cref{proposition:first-second-gap}]
We prove this proposition via example. Let $\exeps' = \exeps / 3$. Consider a market where $\Men = \{\men\}$ and $\Women = \{\women_H, \women_L\}$. That is, there is a single type on side $\Men$, while side $\Women$ is split into agents of high ($H$) and low ($L$) types.  We assume that the utility distributions $\Dist_{\men\women_H}$ and $\Dist_{\men\women_L}$ are point masses at $\frac 1{\exeps'}$ and $\frac 12$, respectively.\footnote{Technically, these distributions do not satisfy our continuity assumption. We present our example this way for simplicity's sake---the same example can be made to work with continuous distributions by adding a small amount of Gaussian noise to each utility value.}
Suppose further that the arrival rates are such that $\alpha_{\men} = 2$, $\alpha_{\women_H} = \exeps'$, and $\alpha_{\women_L} = 2 - \exeps'$. Finally, we consider this example with $\delta = \exeps'$ (i.e., the market becomes nearly frictionless as $\exeps'\to 0$).

It is not hard to see that the first-best payoff involves matching the entire supply of type $\men$ to the entire supplies of types $\women_H$ and $\women_L$, with agents accepting all matches. The first-best social welfare is thus
\[ \frac{2}{1 + \delta}\left(\alpha_{\women_H}\cdot \frac 1{\exeps'} + \alpha_{\women_L}\cdot \frac 12\right) = \frac{4 - \exeps'}{1 + \exeps'}. \]
Notice that $\frac{4-\exeps'}{1+\exeps'}\ge 4 - 5\exeps'$.

On the other hand, consider any stationary equilibrium of this market. We consider two cases: whether or not type $\men$ agents match with type $\women_L$ agents. If type $\men$ agents match with type $\women_L$ agents, then the expected utility of type $\men$ agents must be at most $\frac 12$, in which case the total welfare of the market is at most $2\cdot \alpha_\men\cdot\frac 12 = 2$. On the other hand, suppose type $\men$ does not match with type $\women_L$ agents. We can upper bound the welfare of this outcome by that of the utility of all type $\women_{H}$ agents getting matched, i.e., $2\cdot \alpha_{\women_H}\cdot \frac{1}{\exeps'} = 2$. It follows that the ratio between the first-best welfare and the best possible outcome for $\PlatformOpt$ is at least
\[ \frac{\frac{4 - \exeps'}{1 + \exeps'}}{2}\ge 2 - \frac{5}{2}\exeps' > 2 - \exeps'. \]
\end{proof}

\subsection{Proof of \Cref{proposition:lp}}
\label{sec:proof:lp}

\begin{proof}[Proof of \Cref{proposition:lp}]
We first reparameterize the flow balance constraint by introducing new variables
\[ \glow_{\men\women}\coloneqq\mass_\men\rate_{\men}(\women) = \mass_\women\rate_{\women}(\men) \]
to replace $\rate_\men(\women)$ and $\rate_\women(\men)$.
(For symmetry in notation, we will use $\glow_{\men\women}$ and $\glow_{\women\men}$ interchangeably.)
With these new variables, we may rewrite the objective symmetrically as
\[ 2\sum_{\men\in\Men}\sum_{\women\in\Women} \paren*{\glow_{\men\women} \int_{\thresh_{\men\women}}^\infty\utility\,d\CDF_{\men\women}}. \]
Furthermore, note that $\mass_\type > 0$ by \eqref{eq:opt-stationarity}. Hence we may scale the capacity constraint \eqref{eq:opt-capacity} by $\mass_\type$ and replace all occurrences of $\mass_\type\rate_\type(\typeii)$ with $\glow_{\type\typeii}$ to obtain the equivalent constraint
\begin{equation}\label{eq:opt-capacity-2}
    \mass_\type\ge\sum_{\typeii} \glow_{\type\typeii}
\end{equation}
for all $\type\in\Types$. We similarly replace the matching rate equation \eqref{eq:opt-rate} with
\begin{equation}\label{eq:opt-rate-2}
    \xi_\type\mass_\type = \sum_{\typeii}\paren*{\glow_{\type\typeii} \int_{\thresh_{\type\typeii}}^\infty \,d\CDF_{\type\typeii}}.
\end{equation}

Our next step is to absolve ourselves of the variables $\xi_\type$ and $\mass_\type$. For the former, we can simply merge constraints \eqref{eq:opt-stationarity} and \eqref{eq:opt-rate-2}, as both are equalities involving $\xi_\type\mass_\type$, to get
\begin{equation}\label{eq:opt-rate-stationarity}
    \alpha_\type = \delta\mass_\type + \sum_{\typeii}\paren*{\glow_{\type\typeii} \int_{\thresh_{\type\typeii}}^\infty \,d\CDF_{\type\typeii}}.
\end{equation}
Now, from \eqref{eq:opt-rate-stationarity}, we can get a definition of $\mass_\type$, which we can substitute into the only remaining constraint \eqref{eq:opt-capacity-2} on $\mass_\type$ to get a combined flow balance, capacity, and stationarity constraint
\begin{equation}\label{eq:opt-combined}
    \alpha_\type\ge\sum_{\typeii}\paren*{\glow_{\type\typeii}\paren*{\delta + \int_{\thresh_{\type\typeii}}^\infty \,d\CDF_{\type\typeii}}}
\end{equation}
for each $\type\in\Types$.

To obtain the linear program constraints, we make one final substitution of
\begin{equation}\label{eq:beta}
\beta_{\men\women}\coloneqq\glow_{\men\women}\paren*{\delta + \int_{\thresh_{\men\women}}^\infty \,d\CDF_{\men\women}},
\end{equation}
into \eqref{eq:opt-combined}, which gives us constraints \eqref{eq:lp-1} and \eqref{eq:lp-2} of the linear program. We also substitute $\beta_{\men\women}$ into the objective, which yields the equivalent objective
\begin{equation}
2\sum_{\men\in\Men}\sum_{\women\in\Women}  \paren*{\beta_{\men\women}\cdot\frac{\int_{\thresh_{\men\women}}^\infty\utility\,d\CDF_{\men\women}}{\delta + \int_{\thresh_{\men\women}}^\infty \,d\CDF_{\men\women}}}.
\end{equation}
The only remaining variables that we haven't taken into account are the thresholds $\thresh_{\men\women}$. However, notice that the optimization problem is effectively unconstrained in $\thresh_{\men\women}$---given any choice of $\beta_{\men\women}$ and $\thresh_{\men\women}$, we may choose a $\glow_{\men\women}$ so that the definition \eqref{eq:beta} is satisfied. Thus, we may simply set $\thresh_{\men\women}$ so that its contribution
\begin{equation}\label{eq:rho}
\frac{\int_{\thresh_{\men\women}}^\infty\utility\,d\CDF_{\men\women}}{\delta + \int_{\thresh_{\men\women}}^\infty \,d\CDF_{\men\women}}
\end{equation}
to the objective is maximized. As this is exactly the definition of $\rho_{\men\women}$, we conclude that the linear program \eqref{eq:lp} has the same objective value as the platform's first-best optimization problem.

It remains to prove the second claim in the equivalence between the linear program \eqref{eq:lp} and the platform's first-best optimization problem. For this, we note that given any feasible choice of $\{\beta_{\men\women}\}_{\men\in\Men,\women\in\Women}$, we can reverse all of the substitutions made and obtain corresponding values for $\{\lambda_\type\}_{\type\in\Types}$, $\{\mass_\type\}_{\type\in\Types}$, $\{\xi_\type\}_{\type\in\Types}$, where the variables $\{\thresh_{\men\women}\}_{\men\in\Men,\women\in\Women}$ are set as to maximize \eqref{eq:rho}.
\end{proof}

\subsection{Proof of \Cref{lemma:rho}}

\begin{proof}[Proof of \Cref{lemma:rho}]
    The logarithmic derivative of $A(\thresh)$ is
    \begin{align*}
        \frac{d}{d\thresh}\log\paren*{A(\thresh)}
        &= \frac{\CDF'(\thresh)}{\delta + \int_\thresh^\infty\,d\CDF} - \frac{\thresh\CDF'(\thresh)}{\int_\thresh^\infty \utility\,d\CDF}
        = \paren*{\int_\thresh^\infty (\utility - \thresh)\,d\CDF - \thresh\delta}\frac{F'(\thresh)}{\paren*{\delta + \int_\thresh^\infty\,d\CDF}\paren*{\int_\thresh^\infty\utility\,d\CDF}}.
    \end{align*}
    The sign of this derivative is given by the first term on the right-hand side, since the fraction on the right-hand side is always non-negative. This first term is monotonically decreasing:
    \[ \frac{d}{d\thresh}\paren*{\int_\thresh^\infty (\utility - \thresh)\,d\CDF - \thresh\delta} = -\paren*{\delta + \int_\thresh^\infty\,d\CDF} < 0. \]
    Therefore, the expression is maximized when this term vanishes, i.e., when
    \[ \thresh = \frac{\int_\thresh^\infty\utility\,d\CDF}{\delta + \int_\thresh^\infty \,d\CDF} = \rho. \]
    Moreover, since the derivative is non-negative when $\thresh\le\rho$ and is non-positive when $\thresh\ge\rho$, we have the desired monotonicities when $\thresh\le\rho$ and $\thresh\ge\rho$ as well.
\end{proof}

We also state and prove a generalization of \Cref{lemma:rho} to non-zero thresholds and multiple types on the opposite side. That is, the fixed point and monotonicity properties hold generally:

\begin{lemma}\label{lemma:rho-2}
    Define 
    \[ B(\thresh) = \frac{\sum_{i=1}^\ntypes\paren*{\rate_i\int_{\max(\thresh, \thresh_i)}^\infty \utility\,d\CDF_i}}{\delta\, + \sum_{i=1}^\ntypes\paren*{\rate_i\int_{\max(\thresh, \thresh_i)}^\infty \,d\CDF_i}}, \]
    where $\delta > 0$, $\rate_i > 0$ for all $i$, $\thresh_i\ge 0$ for all $i$, and $\CDF_i$ is a continuous distribution for each $i$. Let $\rho$ satisfies the fixed point equation $\rho = \max_{\thresh\ge 0} B(\thresh)$. Then $B(\rho) = \rho$. Moreover, $B$ is monotonically increasing for $\thresh\le\rho$ and monotonically decreasing for $\thresh\ge\rho$.
\end{lemma}

\begin{proof}
Note that we can write $B(\thresh)$ as
\[ B(\thresh) = \frac{\int_\thresh^\infty \utility\,dG}{\delta + \int_\thresh^\infty\,dG}, \]
where $dG = \sum_{i=1}^\ntypes (\rate_i\cdot\mathbf{1}_{\ge\thresh_i}\cdot d\CDF_i)$ is a new measure defined on $\R$ in terms of $\rate_i$, $\thresh_i$, and $\CDF_i$. Here, $\mathbf{1}_{\ge\thresh_i}$ denotes the indicator function that is $1$ on inputs at least $\thresh_i$ and $0$ otherwise. After appropriately normalizing the numerator and denominator (by $\int_\R\,dG$), the claim follows from \Cref{lemma:rho}.
\end{proof}

\subsection{Proofs of Lemmas from \Cref{sec:step-3}}\label{sec:approximation-proofs-step-3}

\begin{proof}[Proof of \Cref{lemma:first-best}]
    Recall the platform's first-best optimization problem \eqref{eq:lp}. In a star-shaped market, this linear program reduces to a fractional knapsack problem where the platform has a knapsack of size $\alpha_\men$ and wishes to fill it with items for each $\women\in\Women$ of weight $\alpha_\women$ and value $\rho_{\men\women}$. Thus, there exists a first-best solution $\{\beta_{\men\women}^*\}_{\women\in\Women}$ that matches type $\men$ agents exclusively with some subset $\{\women_{1},\ldots,\women_\nfill\}\subseteq\Women$ such that $\rho_{\men\women_{1}}\ge\rho_{\men\women_{2}}\ge\cdots\ge\rho_{\men\women_\nfill}$ and $\beta_{\men\women_i}^* = \alpha_\women$ for all $i < \nfill$. 
\end{proof}

\begin{proof}[Proof of \Cref{lemma:thresh-star}]
    If $i < \nfill$, then $\beta_{\men\women_i}^* = \alpha_{\women_i}$. Therefore,
    \[ \hat\thresh_{\women_i} = \frac{1}{\alpha_{\women_i}} \glow_{\men\women_i}^* \int_{\hat\thresh_{\women_i}}^\infty \utility\,d\CDF_{\men\women_i} = \frac{\int_{\hat\thresh_{\women_i}}^\infty \utility\,d\CDF_{\men\women_i}}{\delta + \int_{\rho_{\men\women_i}} \,d\CDF_{\men\women_i}}. \]
    And by \Cref{lemma:rho}, $\hat\thresh_{\women_i} = \rho_{\men\women_i}$ solves the equation. For $i = \nfill$, since $\beta_{\men\women_\nfill}^*\le\alpha_{\women_\nfill}$, we have
    \[ \hat\thresh_{\women_\nfill} = \frac{1}{\alpha_{\women_\nfill}} \glow_{\men\women_\nfill}^* \int_{\hat\thresh_{\women_\nfill}}^\infty \utility\,d\CDF_{\men\women_\nfill}\le \frac{\int_{\hat\thresh_{\women_\nfill}}^\infty \utility\,d\CDF_{\men\women_\nfill}}{\delta + \int_{\rho_{\men\women_\nfill}} \,d\CDF_{\men\women_\nfill}}. \]
    If $\hat\thresh_{\women_\nfill} > \rho_{\men\women_\nfill}$, then by \Cref{lemma:rho}, we would have that the right-hand side is at most $\rho_{\men\women_\nfill}$, since the expression \smash{$\int_{\thresh}^\infty\utility\,d\CDF_{\men\women_\nfill}$} is monotonically decreasing in $\thresh$. 
    This contradicts the assumption that $\hat\thresh_{\women_\nfill} > \rho_{\men\women_\nfill}$.
    Hence it must be the case that $\hat\thresh_{\women_\nfill}\le\rho_{\men\women_\nfill}$.
\end{proof}

\begin{proof}[Proof of \Cref{lemma:thresh-ub}]
    This follows from the fact that for any threshold $\thresh_\men$, we have that
    \[ \thresh_{\women_i} = \frac{1}{\alpha_\type}\glow_{\men\women_i}^*\int_{\max(\thresh_\men, \thresh_{\women_i})}^\infty \utility\,d\CDF_{\men\women_i}\le\frac{1}{\alpha_{\women_i}}\glow_{\men\women_i}^*\int_{\thresh_{\women_i}}^\infty \utility\,d\CDF_{\men\women_i}. \]
    The argument used in \Cref{lemma:thresh-star} when $i = \nfill$ then lets us conclude that $\thresh_{\women_i}\le\hat\thresh_{\women_i}$ for all $i$.
\end{proof}

\begin{proof}[Proof of \Cref{lemma:thresh-max}]
The lemma is clearly true if $\thresh_\men\ge\hat\thresh_{\women_i}$. Now, suppose $\thresh_{\men} < \hat\thresh_{\women_i}$. Then $\hat\thresh_{\women_i}$ would satisfy the fixed point equation \eqref{eq:fixedpoint}, meaning $\hat\thresh_{\women_i}$ would be the equilibrium threshold of type $\women_i$.
\end{proof}

\section{Hardness of Approximation}\label{sec:hardness}

In this section, we show that the platform's computational problem of finding a $c$-approximate welfare-maximizing stationary equilibrium is $\NP$-hard for some constant $c > 1$. This shows that, conditioned on $\mathsf{P}\neq\NP$, one cannot hope to obtain better than a constant factor approximation to the optimal in polynomial time.

From a technical perspective, our approach very closely mirrors that of \citet{ChakrabartyGoel10}, who show hardness of approximation for related problems (e.g., maximum budgeted allocation). While their technique applies to our setting, it is not clear that their hardness results apply directly---a challenge specific to our setting is that our optimization problem is continuous rather than discrete. To make this approach work, we must show that our complicated feasibility set introduces a discrete element to the optimal allocation. But as a consequence, we obtain a slightly worse constant $c$.

At a high level, we show hardness of approximation by reducing the problem of approximating $\MaxThreeLinTwo$ to finding a sufficiently good approximation for the platform's welfare maximization problem. (Recall that $\MaxThreeLinTwo$ is the optimization problem where, given $\neqs$ linear equations in $\nvars$ variables over $\GF_2$ such that each equation involves exactly three variables, the objective is to maximize the number of equations simultaneously satisfied by an assignment to the variables.) That $\MaxThreeLinTwo$ is hard to approximate is due to \citet{Hastad01} and is a direct corollary of his celebrated $3$-bit PCP. We state this hardness result in terms of a promise problem version of $\MaxThreeLinTwo$:
\begin{theorem}[\citet{Hastad01}]\label{theorem:hastad}
It is $\NP$-hard to distinguish between $\MaxThreeLinTwo$ instances where all but $\approxerr\neqs$ equations are simultaneously satisfiable and where at most $\paren*{\frac 12 + \approxerr}\neqs$ equations are simultaneously satisfiable for any $\approxerr > 0$.
\end{theorem}

From this classical starting point, we derive hardness of approximation for the platform's optimization problem:
\begin{theorem}\label{theorem:hardness}
Computing a $\paren*{\frac{24}{23} - \approxerr}$-approximate solution to the platform's welfare maximization problem is $\NP$-hard for any $\approxerr > 0$.
\end{theorem}

\begin{proof}
As stated above, we reduce instances of $\MaxThreeLinTwo$ to instances of the platform's welfare maximization problem. Specifically, we reduce to instances where the discount rate $\delta$ is small and the distributions $\Dist_{\men\women}$ are point masses.\footnote{Although point masses are not continuous distributions, which we required in \Cref{sec:utilities}, we can make point masses continuous by adding a small amount of (bounded) noise. Doing so does not change any of the conclusions. To simplify our exposition, we focus on the case of point masses only.}

Consider a $\MaxThreeLinTwo$ instance in $\nvars$ variables  $\var_1,\var_2,\ldots,\var_\nvars$ and $\neqs$ equations $\var_{i_\ell}\oplus \var_{j_\ell} \oplus \var_{k_\ell} = \bit_\ell$ for $\ell\in [\neqs]$. We construct an instance of our matching market by defining types, arrival rates, and payoffs based on these variables and equations:
\begin{description}
    \item[Types.] Define \emph{variable types} $\var_{i,0},  \var_{i,1}\in\Men$ and \emph{switch type} $\switch_{i}\in\Women$ for each variable $\var_i$. Define \emph{equation types} $\eqn_{\ell,(0,0,0)}, \eqn_{\ell,(0,1,1)}, \eqn_{\ell,(1,0,1)}, \eqn_{\ell,(1,1,0)}\in\Women$ for each equation $\var_{i_\ell}\oplus \var_{j_\ell} \oplus \var_{k_\ell} = \bit_\ell$, one type for each satisfying assignment to equation $\ell$. That is, type $\eqn_{\ell,(\bit_i,\bit_j,\bit_k)}$ corresponds to the satisfying assignment $(\var_{i_\ell}, \var_{j_\ell}, \var_{k_\ell}) = (\bit_i\oplus\bit_\ell, \bit_j\oplus\bit_\ell, \bit_k\oplus\bit_\ell)$.
    \item[Arrival rates.] Let the variable and switch types for $\var_i$ each arrive at rate $\alpha_i = 4\neqs_i$, where $\neqs_i$ is the number of equations that involve $\var_i$. Let each equation type arrive at rate $3$.
    \item[Payoffs.] Let matching type $\var_{i,\bit}$ and type $\switch_i$ yield payoff $\switchval$ for all $\bit$. Let matching type $\var_{i,\bit}$ with any equation type whose corresponding satisfying assignment has the variable $\var_i$ with parity $\bit$ yield payoff $1$. All other matches provide no utility.
\end{description}
With this setup, no agent on side $\Women$ will ever reject positive utility matches in equilibrium, since all positive utility matches yield the same payoff for them. Moreover, no switch type agent will ever be rejected in equilibrium, since they yield maximal payoff for variable type agents.\footnote{When $\delta > 0$, these observations still hold for distributions that are not point masses, so long as the variation in payoff is sufficiently small.} And if in equilibrium an agent of type $\var_{i,\bit}$ has expected continuation utility greater than $1$, then they must only accept agents of type $\switch_i$. These observations set the stage for our key lemma:

\begin{lemma}
In any welfare-maximizing equilibrium, the expected payoff of all switch type agents will be $\switchval / (1 + \delta)$. Furthermore, such an equilibrium can be achieved by only showing type $\switch_i$ and type $\var_{i,\bit}$ agents to each other for some $\bit\in\{0, 1\}$ for each $i$.
\end{lemma}

\begin{proof}
First, we show that some variable type $\var_{i,\bit}$ must have expected payoff greater than $1$ for each $i\in[\nvars]$. Suppose otherwise, that there exists a welfare-maximizing stationary equilibrium such that the expected payoffs of both variable types $\var_{i,\bit}$ are at most $1$ for some $i$. Then we could stop matching both types entirely, by setting their assortments to $0$. This comes at a loss in welfare of at most $2\cdot 4\neqs_i$. Then, we could match type $\var_{i,0}$ entirely to $\switch_i$, and obtain an increase in welfare of $4\neqs_i\cdot\switchval / (1 + \delta) > 2\cdot 4\neqs_i$. Finally, note that these new assortments also produce a stationary equilibrium: Unmatching the variable types $\var_{i,\bit}$ only adds slack to the assortment feasibility constraints. Unmatching also does not affect any equation type's threshold, since they always accept in the first place. This new equilibrium has higher welfare, which contradicts our premise that the original equilibrium is welfare-maximizing.

The preceding argument shows that in any welfare-maximizing equilibrium, there is a variable type $\var_{i,\bit}$ whose expected utility is greater than $1$ for each $i$. Notice that this type must exclusively match with type $\switch_i$ agents, since equation type agents only yield payoff $1$ and thus cannot be accepted in equilibrium. 
To prove the second part of this lemma, we show that we can obtain another welfare-maximizing equilibrium by matching the entire supply of type $\switch_i$ agents to this variable type $\var_{i,\bit}$.

Suppose we have a welfare-maximizing equilibrium, and we modify it by matching the entire supply of type $\switch_i$ agents to a variable type $\var_{i,\bit}$ whose expected utility exceeds $1$. To restore a stationary equilibrium, note that making this modification may violate two sets of constraints: First, it might violate type $\var_{i,\bit}$'s feasibility constraint; but this is not a problem, since we can simply stop showing $\var_{i,\bit}$ type agents any types but type $\switch_i$, since these agents were only accepting type $\switch_i$ agents in the first place. Second, it might cause type $\var_{i,\bit\oplus 1}$ to want to accept more equation type agents because they stop matching with switch type agents entirely. 
To resolve this, we can simply reduce the supply of equation type agents to type $\var_{i,\bit\oplus 1}$ so that by accepting all such agents, $\var_{i,\bit\oplus 1}$ accepts the same quantity of agents as in the original equilibrium. Finally, to see that this new equilibrium has the same welfare as the original equilibrium, note that the flow rates of matches for both switch and equation types are at least as high as before.
\end{proof}

This lemma tells us that there exists a welfare-maximizing equilibrium where, for every $i$, one of variable types $\var_{i,0}$ and $\var_{i,1}$ is fully matched to the switch type $\switch_i$. Moreover, it suffices to optimize over such equilibria. Notice that each such equilibrium also corresponds to an assignment $(\var_1,\ldots,\var_\nvars)\in\{0,1\}^\nvars$, where $\var_i = \bit$ if type $\var_{i,\bit\oplus 1}$ is fully matched to the switch type $\switch_i$.

We are now ready to state our analogs of Lemmas 4.6 and 4.7 of \citet{ChakrabartyGoel10}:

\begin{lemma}\label{lemma:complete}
Given a $\MaxThreeLinTwo$ instance with $\neqsii\ge (1 - \approxerr)\neqs$ simultaneously satisfiable equations, the platform's problem it maps to has $\OPT\ge (36 + 48\delta - 12\approxerr)\neqs / (1 + \delta)$.
\end{lemma}

\begin{proof}
Given an assignment $(\var_1,\var_2,\ldots,\var_\nvars)\in\{0,1\}^\nvars$ that satisfies $\neqsii$ equations, we can match the switch types to the variable types as above. We can then match the remaining variable types to the equation types so that each variable type matches with the equation types of each satisfied equation it is part of at a flow rate of $4 / (1 + \delta)$. 
We thus have payoff $2(1+2\delta) / (1 + \delta)\cdot \sum_i 4\neqs_i$ from switch types and payoff $12/(1+\delta)$ per satisfied equation. Summing these quantities gives us the lemma.
\end{proof}

\begin{lemma}\label{lemma:sound}
Given a $\MaxThreeLinTwo$ instance with $\neqsii\le \paren*{\frac 12 + \approxerr}\neqs$ simultaneously satisfiable equations, the platform's problem it maps to has $\OPT\le (34.5 + 48\delta + 3\approxerr)\neqs / (1 + \delta)$.
\end{lemma}

\begin{proof}
Take any assignment $(\var_1,\var_2,\ldots,\var_\nvars)\in\{0,1\}^\nvars$ and match switch types to variable types as described above. Now, for any equation that is not satisfied, notice that there must be some equation type that cannot be matched at all to any of the remaining variable types. So the payoff from the equation types for unsatisfied equations is at most $9 / (1 + \delta)$. Upper bounding the payoffs of the equation types of the satisfied equations by $12/(1+\delta)$ per satisfied equation, we get $12\neqs\cdot2(1+2\delta)/(1+\delta)$ from switch types and at most payoff $((\frac 12 - \approxerr)\neqs \cdot 9+ (\frac 12 + \approxerr)\neqs\cdot 12)/(1+\delta)$ from equation types. Summing these quantities proves the lemma.
\end{proof}

To finish, observe that \Cref{lemma:complete,lemma:sound} show that if we can approximate the platform's welfare-maximization problem with a factor of
\[ \frac{36 + 48\delta - 12\approxerr}{34.5 + 48\delta + 3\approxerr}, \]
then we would be able to distinguish the two classes of $\MaxThreeLinTwo$ inputs in \Cref{theorem:hastad}. Thus, by taking sufficiently small $\delta$ and $\approxerr$, we see that it is $\NP$-hard to approximate the platform's welfare-maximization problem within any constant factor less than $\frac{36}{34.5} = \frac{24}{23}$.
\end{proof}

\end{APPENDIX}

\end{document}